\documentclass[journal]{IEEEtran}

\usepackage[utf8]{inputenc}
\usepackage{color}
\usepackage{xcolor}
\usepackage{array}
\usepackage{verbatim}
\usepackage{float}
\usepackage{amsmath}
\usepackage{amsthm}
\usepackage{amssymb}
\usepackage{graphicx}
\usepackage{longtable}
\usepackage{multirow}
\usepackage{booktabs}
\usepackage[unicode=true,
bookmarks=false,
breaklinks=false,pdfborder={0 0 1},colorlinks=false]
{hyperref}
\hypersetup{
	colorlinks,bookmarksopen,bookmarksnumbered,citecolor=blue,urlcolor=blue}
\usepackage{cite}

\usepackage{lipsum}
\usepackage{mathtools}
\usepackage{cuted}
\providecommand{\tabularnewline}{\\}
\usepackage{algorithmic}
\usepackage{longtable}
\usepackage{balance}

\floatstyle{ruled}
\newfloat{algorithm}{tbp}{loa}
\providecommand{\algorithmname}{Algorithm}
\floatname{algorithm}{\protect\algorithmname}

\makeatletter
\let\oldforeign@language\foreign@language
\DeclareRobustCommand{\foreign@language}[1]{%
	\lowercase{\oldforeign@language{#1}}}

\let\oldforeign@language\foreign@language
\DeclareRobustCommand{\foreign@language}[1]{%
	\lowercase{\oldforeign@language{#1}}}

\ifCLASSINFOpdf
\else
\fi

\newcommand{\MYfooter}{\smash{
		\hfil\parbox[t][\height][t]{\textwidth}{\centering
			\thepage}\hfil\hbox{}}}

\def\ps@IEEEtitlepagestyle{%
	\def\@oddhead{\parbox[t][\height][t]{\textwidth}{\centering \scriptsize
			Personal use of this material is permitted. Permission from the author(s) and/or copyright holder(s), must be obtained for all other uses. Please contact us and provide details if you believe this document breaches copyrights.\\
			\noindent\makebox[\linewidth]{}
		}\hfil\hbox{}}%
	\def\@evenhead{\scriptsize\thepage \hfil \leftmark\mbox{}}%
	\def\@oddfoot{\parbox[t][\height][l]{\textwidth}{
			\vspace{-20pt}{\rule{\textwidth}{0.4pt}}\\ \footnotesize\underline{To cite this article:}
			{\bf{\footnotesize\textcolor{red}{H. A. Hashim and K. G. Vamvoudakis, "Adaptive Neural Network Stochastic-Filter-based Controller for Attitude Tracking with Disturbance Rejection," IEEE Transactions on Neural Networks and Learning Systems, vol. PP, no. PP, pp. 1-11, 2022.}}} doi: \href{https://doi.org/10.1109/TNNLS.2022.3183026}{10.1109/TNNLS.2022.3183026}\\
			\noindent\makebox[\linewidth]
		}\hfil\hbox{}}%
	\def\@evenfoot{\MYfooter}}

\makeatother
\newtheorem{defn}{Definition}

\newtheorem{lem}{Lemma}

\newtheorem{thm}{Theorem}

\newtheorem{assum}{Assumption}

\begin{document}
	\bstctlcite{IEEEexample:BSTcontrol}

	\title{Adaptive Neural Network Stochastic-Filter-based Controller for Attitude Tracking with\\ Disturbance Rejection}

\author{Hashim A. Hashim and Kyriakos G. Vamvoudakis
	\thanks{This work was supported in part by National Sciences and Engineering Research Council of Canada (NSERC), under the grants RGPIN-2022-04937 and DGECR-2022-00103, and by NSF under grant Nos. CAREER CPS-$1851588$, CPS-$2038589$, and S\&AS-$1849198$ and by NASA ULI under grant number $80$NSSC$20$M$0161$.}
	\thanks{$^*$Corresponding author, H. A. Hashim is with the Department of Mechanical and Aerospace Engineering, Carleton University, Ottawa, ON, K1S 5B6, Canada (e-mail: hhashim@carleton.ca)}
	\thanks{K. G. Vamvoudakis is with the Daniel Guggenheim School of Aerospace Engineering, Georgia Institute of Technology, Atlanta, GA, 30332, USA e-mail: kyriakos@gatech.edu}

}



\maketitle

\begin{abstract}
This paper proposes a real-time neural network (NN) stochastic filter-based
controller on the Lie Group of the Special Orthogonal Group $SO(3)$
as a novel approach to the attitude tracking problem. The introduced
solution consists of two parts: a filter and a controller. Firstly,
an adaptive NN-based stochastic filter is proposed that estimates
attitude components and dynamics using measurements supplied by onboard
sensors directly. The filter design accounts for measurement uncertainties
inherent to the attitude dynamics, namely unknown bias and noise corrupting
angular velocity measurements. The closed loop signals of the proposed
NN-based stochastic filter have been shown to be semi-globally uniformly
ultimately bounded (SGUUB). Secondly, a novel control law on $SO(3)$ coupled with the proposed
estimator is presented. The control law addresses unknown disturbances.
In addition, the closed loop signals of the proposed filter-based
controller have been shown to be SGUUB. The proposed approach offers
robust tracking performance by supplying the required control signal
given data extracted from low-cost inertial measurement units. While
the filter-based controller is presented in continuous form, the discrete
implementation is also presented. Additionally, the unit-quaternion
form of the proposed approach is given. The effectiveness and robustness
of the proposed filter-based controller is demonstrated using its
discrete form and considering low sampling rate, high initialization
error, high-level of measurement uncertainties, and unknown disturbances.
\end{abstract}


\IEEEpeerreviewmaketitle{}

\section{Introduction}

\IEEEPARstart{A}{ttitude} estimation and tracking control of a rigid-body
rotating in three-dimensional (3D) space are indispensable tasks for
majority of robotics and aerospace applications. Examples include
satellites, rotating radars, unmanned aerial vehicles (UAVs), space
telescopes, to name a few \cite{zlotnik2016nonlinear,hashim2022ExpVTOL,hu2017observer,hashim2019SO3Wiley,chai2019six,lee2012robust}.
The research in area of attitude estimation and control has made a
great leap forward with the introduction of micro-electro-mechanical
systems (MEMS) \cite{kuo2008open} that allowed for the design of
compact low-cost onboard units such as inertial measurement units
(IMUs) and magnetic, angular rate, and gravity (MARG) sensors. Although
low-cost sensors are inexpensive, low-weight, compact, and power-efficient,
they supply imperfect measurements corrupted with unknown bias and
noise \cite{pham2015gain,hashim2019SO3Wiley,zlotnik2016nonlinear}.
Controlling rigid-body's rotation in 3D space requires the knowledge
of the true attitude. Unfortunately, the true attitude is unknown
and has to be obtained through estimation or algebraic reconstruction
using, for instance IMU or MARG sensors. The algebraic attitude reconstruction
involves using number of measurements and an algebraic algorithm,
for instance QUEST algorithm \cite{shuster1981three} or singular
value decomposition (SVD) \cite{markley1988attitude}. However, algebraic
reconstruction is ineffective when sensor measurements are heavily
contaminated by uncertainties. As such, estimation approaches that
use filtering techniques have acquired great importance \cite{choukroun2006novel,pham2015gain,markley2003attitude,crassidis2003unscented,hashim2019SO3Wiley}.

Conventionally, attitude estimation has been addressed considering
Kalman-type filters \cite{hashim2018SO3Stochastic}. Examples include Kalman filter (KF) \cite{choukroun2006novel},
extended Kalman filter (EKF) \cite{pham2015gain}, multiplicative
EKF (MEKF) \cite{markley2003attitude}, unscented Kalman filter (UKF)
\cite{crassidis2003unscented}, and invariant EKF (IEKF) \cite{phogat2020invariant}.
The filters in \cite{choukroun2006novel,pham2015gain,markley2003attitude,crassidis2003unscented}
are unit-quaternion based which offers nonsingularity of attitude
representation, but is challenged by nonuniqueness \cite{shuster1993survey}.
To address the nonuniqueness limitation of unit-quaternion, a set
of nonlinear attitude observing/filtering algorithms on the Special
Orthogonal Group $\mathbb{SO}(3)$ have been proposed \cite{mahony2008nonlinear,zlotnik2016nonlinear,hashim2019SO3Wiley,hashim2018SO3Stochastic}.
$\mathbb{SO}(3)$ provides global and unique attitude representation
\cite{shuster1993survey}. Moreover, in comparison with Kalman-type
filters, nonlinear attitude filtering solutions on $\mathbb{SO}(3)$
have been shown to be 1) simpler in design, 2) computationally cheap,
and 3) better in terms of tracking performance \cite{mahony2008nonlinear,zlotnik2016nonlinear,hashim2019SO3Wiley}.
With regard to attitude control, over the last two decades multiple
successful control strategies have been introduced to control rigid-body's
attitude given accurate attitude and angular velocity information
\cite{zhu2011adaptive,zou2010quaternion,tian2015finite}. Other solutions
developed attitude tracking control schemes reliant on the true attitude
information without angular velocity measurements \cite{zlotnik2014rotation,xiao2014fault}.
In the above-discussed literature, attitude estimation and tracking
control solutions are designed separately. However, the attitude problem
is highly nonlinear and coupling a standalone filter design with a
standalone controller design cannot guarantee overall stability, especially
if 1) the rigid-body is equipped with a low-cost measurement unit
and 2) a large initialization error is present between the true and
the estimated states \cite{hashim2022ExpVTOL}. 

With an objective of overcoming the overall stability challenge, several
state-of-the-art observer-based controllers have been proposed, such
as, an observer-based controller guaranteeing local exponential stability
\cite{caccavale1999output}, a full-state observer-based controller
for rigid-body motion \cite{salcudean1991globally}, a hybrid control
scheme ensuring semi-global asymptotic stability reliant on a switching
observer for restoring angular velocity data \cite{mayhew2011quaternion},
observer-based controller with finite time convergence \cite{hu2017observer},
and observer-based controller for unknown exterior disturbances \cite{liu2018disturbance}.
The limitations of the existing observer-based controller solutions,
such as \cite{hu2017observer,caccavale1999output,salcudean1991globally,mayhew2011quaternion,liu2018disturbance}
are three-fold: 1) for the sake of simplicity these techniques disregard
uncertainties of the onboard sensing units in the stability analysis,
2) they rely on reconstructed attitude which increases the computational
cost, and 3) they only consider the case of known nonlinear dynamics
of the attitude problem. However, in practice, affordable systems
are likely to 1) be equipped with low-cost sensing units, 2) operate
in uncertain environments where the model dynamics may not be accurately
known, and 3) be affected by unknown disturbances.

Neural networks (NNs) are known to be a powerful tool for learning
and estimating complex nonlinear systems \cite{chai2019six}. Over
the past few years, adaptive NNs have been shown to be efficient for
online estimation of unknown high-order nonlinear dynamics. Successful
applications of NNs include but are not limited to nonlinear attitude filtering \cite{hashim2022NNFilter},  adaptive consensus of networked systems \cite{zheng2017consensus}, trajectory tracking
of ground robots \cite{chen2017robust}, stochastic nonlinear systems \cite{li2022command, hashim2022NNFilter},
and strict-feedback multi-input multi-output systems \cite{li2019adaptive}.
Accurate NN estimation of unknown high-order nonlinear dynamics results
in a successful control process \cite{chai2019six,chen2017robust,zheng2017consensus}.
With an objective of addressing the aforementioned shortcomings of
the observer-based controllers \cite{mayhew2011quaternion,caccavale1999output,salcudean1991globally,hu2017observer,liu2018disturbance},
in this paper, the nonlinear attitude problem is modelled on the Lie
Group of $\mathbb{SO}(3)$ that offers global and unique attitude
representation. The unknown bias and noise corrupting angular velocity
measurements are tackled by incorporating stochastic differential
equations in problem formulation. The measurement uncertainties, unknown
disturbances, and overall stability are addressed through proposing
NN stochastic filter-based controller for the attitude tracking problem.
The contributions of this paper are as follows:
\begin{enumerate}
	\item[1)] A real-time NN-based nonlinear stochastic filter able to use available
	measurements directly for the attitude estimation problem is proposed
	on $\mathbb{SO}(3)$;
	\item[2)] The proposed NN-based filter considers unknown random noise as well
	as constant bias corrupting the velocity measurement;
	\item[3)] Using Lyapunov stability, the closed loop error signals of the stochastic
	filter design are guaranteed to be semi-globally uniformly ultimately
	bounded (SGUUB) in mean square;
	\item[4)] A control law has been proposed considering states estimated by the
	NN-based filter, uncertain measurements, and unknown disturbances;
	\item[5)] The overall stability has been proven, and the closed loop error
	signals of the filter-based controller have been shown to be SGUUB
	using Lyapunov stability;
\end{enumerate}
To the best of our knowledge, the attitude tracking problem has not
been previously addressed using NN-based nonlinear stochastic filter-based
controller on $\mathbb{SO}(3)$. The proposed approach is robust against
disturbances and can be a prefect fit for systems equipped with low-cost
sensing units. Furthermore, the proposed approach provides strong
tracking performance, is computationally cheap, and has been tested
in its discrete form at a low sampling rate.

The paper is organized to contain eight Sections. Section \ref{sec:Preliminaries}
introduces preliminaries and math notation related to attitude and
$\mathbb{SO}(3)$. Section \ref{sec:Problem-Formulation} outlines
the attitude problem, presents the sensor measurements, error criteria,
and formulates the problem with respect to stochastic differential
equations. Section \ref{sec:NN-Filter} presents NN approximation
of the nonlinear attitude problem and proposes a novel NN-based stochastic
attitude filter. Section \ref{sec:Controller} introduces the control
law and the innovation terms. Section \ref{sec:Summary-of-Implementation}
presents a summary of the implementation in a discrete form. Section
\ref{sec:SE3_Simulations} demonstrates the numerical results. Section
\ref{sec:SE3_Conclusion} concludes the work.

\section{Preliminaries\label{sec:Preliminaries}}

$\mathbb{R}$ and $\mathbb{R}_{+}$ refer to a set of real numbers
and nonnegative real numbers, respectively, while $\mathbb{R}^{n\times m}$
stands for an $n$-by-$m$ dimensional space. $\mathbf{I}_{n}\in\mathbb{R}^{n\times n}$
describes an identity matrix, and $0_{n\times m}\in\mathbb{R}^{n\times m}$
represents a matrix of zeros. $||u||=\sqrt{u^{\top}u}$ denotes the
Euclidean norm of $u\in\mathbb{R}^{n}$. $||W||_{F}=\sqrt{{\rm Tr}\{WW^{*}\}}$
refers to the Frobenius norm of matrix $W\in\mathbb{R}^{q\times m}$
with $*$ standing for a conjugate transpose. $\exp(\cdot)$, $\mathbb{P}\{\cdot\}$,
and $\mathbb{E}[\cdot]$ refers to an exponential, a probability,
and an expected value of a component. Let us define $A\in\mathbb{R}^{n\times n}$
with $\lambda(A)=\{\lambda_{1},\lambda_{2},\ldots,\lambda_{n}\}$
being a set of eigenvalues. $\overline{\lambda}_{A}=\overline{\lambda}(A)$
represents the maximum value, and $\underline{\lambda}_{A}=\underline{\lambda}(A)$
stands for the minimum value of $\lambda(A)$. For simplicity, $\left\{ \mathcal{I}\right\} $
denotes a fixed inertial-frame and $\left\{ \mathcal{B}\right\} $
stands for a fixed body-frame. Orientation of a rigid-body is known
as attitude $R\in\mathbb{SO}(3)$ described by
\[
\mathbb{SO}(3)=\{R\in\mathbb{R}^{3\times3}|R^{\top}R=\mathbf{I}_{3}\text{, }{\rm det}(R)=+1\}
\]
with ${\rm det}(\cdot)$ denoting a determinant. $\mathfrak{so}(3)$
describes the Lie algebra of $\mathbb{SO}(3)$ given by
\begin{align*}
	\mathfrak{so}(3) & =\{[u]_{\times}\in\mathbb{R}^{3\times3}|[u]_{\times}^{\top}=-[u]_{\times},u\in\mathbb{R}^{3}\}\\{}
	[u]_{\times} & =\left[\begin{array}{ccc}
		0 & -u_{3} & u_{2}\\
		u_{3} & 0 & -u_{1}\\
		-u_{2} & u_{1} & 0
	\end{array}\right]\in\mathfrak{so}\left(3\right),\hspace{1em}u=\left[\begin{array}{c}
		u_{1}\\
		u_{2}\\
		u_{3}
	\end{array}\right]
\end{align*}
$\mathbf{vex}$ defines the inverse mapping of $[\cdot]_{\times}$
where $\mathbf{vex}:\mathfrak{so}(3)\rightarrow\mathbb{R}^{3}$ such
that $\mathbf{vex}([u]_{\times})=u,\forall u\in\mathbb{R}^{3}$ while
$\boldsymbol{\mathcal{P}}_{a}:\mathbb{R}^{3\times3}\rightarrow\mathfrak{so}(3)$
is an anti-symmetric projection operator where
\[
\boldsymbol{\mathcal{P}}_{a}(A)=\frac{1}{2}(A-A^{\top})\in\mathfrak{so}\left(3\right),\forall A\in\mathbb{R}^{3\times3}
\]
Consider $A=[a_{i,j}]_{i,j=1,2,3}\in\mathbb{R}^{3\times3}$ and define
\begin{equation}
	\boldsymbol{\Upsilon}(A)=\mathbf{vex}(\boldsymbol{\mathcal{P}}_{a}(A))=\frac{1}{2}\left[\begin{array}{c}
		a_{32}-a_{23}\\
		a_{13}-a_{31}\\
		a_{21}-a_{12}
	\end{array}\right]\in\mathbb{R}^{3}\label{eq:Attit_VEX}
\end{equation}
Define the Euclidean distance of $R$
\begin{equation}
	||R||_{{\rm I}}=\frac{1}{4}{\rm Tr}\{\mathbf{I}_{3}-R\}\in\left[0,1\right],\hspace{1em}R\in\mathbb{SO}(3)\label{eq:Attit_Ecul_Dist}
\end{equation}
with ${\rm Tr}\{\cdot\}$ referring to a trace of a matrix (visit
\cite{hashim2019SO3Wiley,hashim2018SO3Stochastic}). Let $R\in\mathbb{SO}(3)$,
$Z\in\mathbb{R}^{3\times3}$, $y,z\in\mathbb{R}^{3}$ and recall the
composition mapping in \eqref{eq:Attit_VEX}. The identities below
hold true and will be useful in the subsequent derivations:
\begin{align}
	z\times y & =yz^{\top}-zy^{\top}\label{eq:Attit_Identity1}\\
	R[y]_{\times}R^{\top} & =[Ry]_{\times}\label{eq:Attit_Identity2}\\
	{\rm Tr}\{Z[y]_{\times}\} & ={\rm Tr}\{\boldsymbol{\mathcal{P}}_{a}(Z)[y]_{\times}\}=-2\boldsymbol{\Upsilon}(Z)^{\top}y\label{eq:Attit_Identity3}
\end{align}

\section{Problem Formulation\label{sec:Problem-Formulation}}

\subsection{Attitude Dynamics and Measurements}

Consider $R\in\mathbb{SO}(3)$ as the rigid-body's attitude and $\Omega$
as the rigid-body's angular velocity where $R,\Omega\in\{\mathcal{B}\}$.
The true attitude and angular velocity dynamics are given by
\begin{equation}
	\begin{cases}
		\dot{R} & =R[\Omega]_{\times}\\
		J\dot{\Omega} & =[J\Omega]_{\times}\Omega+\mathcal{T}+d
	\end{cases}\label{eq:Attit_R_dot}
\end{equation}
where $J=J^{\top}\in\mathbb{R}^{3\times3}$ denotes the rigid-body's
inertia matrix (positive-definite), $\mathcal{T}\in\mathbb{R}^{3}$
stands for the rotational torque (control input signal), and $d\in\mathbb{R}^{3}$
denotes an unknown constant disturbance vector with $J,\mathcal{T},d\in\{\mathcal{B}\}$.
The attitude can be defined using a set of $N$ observations in $\{\mathcal{I}\}$
and $N$ respective measurements in $\{\mathcal{B}\}$. Note that,
at least two non-collinear observations and measurements must be available.
Examples of common low-cost units for attitude determination and estimation
include \cite{hashim2018SO3Stochastic,choukroun2006novel,zlotnik2016nonlinear,hashim2019SO3Wiley,grip2011attitude}: 
\begin{itemize}
	\item An inertial measurement unit (IMU) composed of a gyroscope (supplies
	angular velocity measurements), a magnetometer (supplies direction
	of the Earth's magnetic field), and an accelerometer (supplies apparent
	acceleration measurements) or
	\item A magnetic, angular rate, and gravity (MARG) sensor.
\end{itemize}
Define $r_{i}\in\mathbb{R}^{3}$ as the $i$th observation in $\{\mathcal{I}\}$
and $y_{i}\in\mathbb{R}^{3}$ as the $i$th measurement in $\{\mathcal{B}\}$
for all $i=1,2,\ldots,N$. The $i$th measurement $y_{i}$ is defined
by \cite{hashim2018SO3Stochastic,hashim2019SO3Wiley,grip2011attitude,hashim2021_COMP_ENG_PRAC}
\begin{align}
	y_{i} & =R^{\top}r_{i}+b_{i}+n_{i}\in\mathbb{R}^{3},\hspace{1em}\forall i=1,2,\ldots,N\label{eq:Attit_Vec_yi}
\end{align}
where $b_{i}$ is unknown constant bias and $n_{i}$ refers to unknown
noise. The expression in \eqref{eq:Attit_Vec_yi} exemplifies measurements
supplied by an IMU such as a magnetometer and an accelerometer. It
is a common approach to normalize inertial-frame observations and
body-frame measurements as follows:
\begin{equation}
	{\bf r}_{i}=\frac{r_{i}}{||r_{i}||},\hspace{1em}{\bf y}_{i}=\frac{y_{i}}{||y_{i}||}\label{eq:Attit_Vec_yi_norm}
\end{equation}
An angular velocity measurement is given by:
\begin{equation}
	\Omega_{m}=\Omega+W_{b}+n\in\mathbb{R}^{3}\label{eq:Attit_Om_m}
\end{equation}
where $\Omega$ refers to the true angular velocity, while $W_{b}$
and $n$ describe unknown weighted bias (constant) and noise.

\subsection{Stochastic Reformulation}

In \eqref{eq:Attit_Om_m} $n$ is a bounded Gaussian noise vector
with $\mathbb{E}[n]=0$. Considering the fact that a derivative of
a Gaussian process leads to a Gaussian process \cite{khasminskii1980stochastic,ito1984lectures,hashim2019SO3Wiley},
one can re-express the noise vector $n$ in terms of Brownian motion
process as follows:
\begin{equation}
	n=\mathcal{Q}\frac{d\beta}{dt}\label{eq:Attit_n_beta}
\end{equation}
where $\beta\in\mathbb{R}^{3}$, and $\mathcal{Q}=\mathcal{Q}^{\top}\in\mathbb{R}^{3\times3}$
refers to an unknown weighted matrix with $\mathcal{Q}^{2}=\mathcal{Q}\mathcal{Q}^{\top}$
being the noise covariance. Note that $\mathbb{P}\{\beta(0)=0\}=1$
and $\mathbb{E}[\beta]=0$ \cite{khasminskii1980stochastic}. Hence,
from \eqref{eq:Attit_Om_m} and \eqref{eq:Attit_n_beta}, the upper
portion of the true attitude dynamics in \eqref{eq:Attit_R_dot} are
re-expressed in a stochastic form as below:
\begin{equation}
	dR=R[\Omega_{m}-W_{b}]_{\times}dt-R[\mathcal{Q}d\beta]_{\times}\label{eq:Attit_dR_dt}
\end{equation}
According to \eqref{eq:Attit_VEX}-\eqref{eq:Attit_Identity3}, it
becomes apparent that the Euclidean distance of the attitude stochastic
dynamics in \eqref{eq:Attit_dR_dt} is equivalent to
\begin{equation}
	d||R||_{{\rm I}}=\frac{1}{2}\boldsymbol{\Upsilon}(R)^{\top}(\Omega_{m}-b)dt-\frac{1}{2}\boldsymbol{\Upsilon}(R)^{\top}\mathcal{Q}d\beta\label{eq:Attit_dR_norm}
\end{equation}

\begin{defn}
	\label{def:Def_SGUUB}\cite{ji2006adaptive,hashim2019SO3Wiley} Consider
	the dynamics in \eqref{eq:Attit_dR_norm} and let $t_{in}$ be an
	initial time. $||R||_{{\rm I}}=||R(t)||_{{\rm I}}$ is defined to
	be almost SGUUB if for a known set $S_{\varrho}\in\mathbb{R}$ and
	$||R(t_{in})||_{{\rm I}}$ there exists a constant $c>0$ and a time
	constant $\tau_{c}=\tau_{c}(\eta,||R(t_{in})||_{{\rm I}})$ with $\mathbb{E}[||R(t_{in})||_{{\rm I}}]<c,\forall t>t_{in}+c$.%
\end{defn}
\begin{lem}
	\label{Lemm:Def_LV_dot}\cite{deng2001stabilization} Consider the
	dynamics in \eqref{eq:Attit_dR_norm} and let $\mathcal{U}(||R||_{{\rm I}})$
	be a twice differentiable cost function with the following differential
	operator:
	\begin{equation}
		\mathcal{L}\mathcal{U}(||R||_{{\rm I}})=\mathcal{U}_{1}^{\top}f+\frac{1}{2}{\rm Tr}\{gg^{\top}\mathcal{U}_{2}\}\label{eq:Attit_Vfunction_Lyap0}
	\end{equation}
	where $f=\frac{1}{2}\boldsymbol{\Upsilon}(R)^{\top}(\Omega_{m}-W_{b})\in\mathbb{R}$,
	$g=-\frac{1}{2}\boldsymbol{\Upsilon}(R)^{\top}\mathcal{Q}\in\mathbb{R}^{1\times3}$,
	$\mathcal{U}_{1}=\partial\mathcal{U}/\partial||R||_{{\rm I}}$, and
	$\mathcal{U}_{2}=\partial^{2}\mathcal{U}/\partial||R||_{{\rm I}}^{2}$.
	Define $\underline{\rho}_{1}(\cdot)$ and $\overline{\rho}_{2}(\cdot)$
	as class $\mathcal{K}_{\infty}$ functions, and let $\mu_{1}>0$ and
	$\mu_{2}\geq0$ where
	\begin{align}
		\underline{\rho}_{1}(||R||_{{\rm I}}) & \leq\mathcal{U}(||R||_{{\rm I}})\leq\overline{\rho}_{2}(||R||_{{\rm I}})\label{eq:Attit_Vfunction_Lyap}\\
		\mathcal{L}\mathcal{U}(||R||_{{\rm I}}) & =\mathcal{U}_{1}^{\top}f+\frac{1}{2}{\rm Tr}\{gg^{\top}\mathcal{U}_{2}\}\nonumber \\
		& \leq-\mu_{1}\mathcal{U}(||R||_{{\rm I}})+\mu_{2}\label{eq:Attit_dVfunction_Lyap}
	\end{align}
	Thereby, the dynamics in \eqref{eq:Attit_dR_norm} have an almost
	unique strong solution on $[0,\infty)$. Additionally, the solution
	$||R||_{{\rm I}}$ is bounded in probability by
	\begin{equation}
		\mathbb{E}[\mathcal{U}(||R||_{{\rm I}})]\leq\mathcal{U}(||R(0)||_{{\rm I}}){\rm exp}(-\mu_{1}t)+\mu_{2}/\mu_{1}\label{eq:Attit_EVfunction_Lyap}
	\end{equation}
	Moreover, the expression in \eqref{eq:Attit_EVfunction_Lyap} indicates
	that $||R||_{{\rm I}}$ is SGUUB.
\end{lem}
\begin{lem}
	\label{lem:Lem_VEX_RMI}\cite{hashim2019SO3Wiley} Let $R\in\mathbb{SO}\left(3\right)$,
	$M=M^{\top}\in\mathbb{R}^{3\times3}$ with $rank(M)\geq2$, and $\overline{M}={\rm Tr}\{M\}\mathbf{I}_{3}-M$.
	Then, the following definitions hold:
	\begin{equation}
		\underline{\lambda}_{\overline{M}}^{2}||R||_{{\rm I}}\leq||\boldsymbol{\Upsilon}(MR)||^{2}\leq\overline{\lambda}_{\overline{M}}^{2}||R||_{{\rm I}}\label{eq:Attit_VEX_RMI}
	\end{equation}
\end{lem}
This work aims to develop a real-time NN stochastic filter-based controller
for the attitude tracking problem of a rigid-body in 3D space. An
adaptive NN-based stochastic filter on $SO(3)$ will be proposed for
estimating the attitude components, dynamics, and angular velocity
uncertainties (unknown constant covariance and bias) using onboard
sensor measurements directly. NN weights will be updated online and
the adaptation mechanism will be extracted using a novel Lyapunov
function candidate. Next, a control law on $SO(3)$ interconnected
with the proposed filter will be presented. The control law will account
for unknown disturbances affecting the rigid-body. Fig. \ref{fig:Graph_Summary}
presents a conceptual summary of the proposed methodology.

\begin{figure}[h]
	\centering{}\includegraphics[scale=0.22]{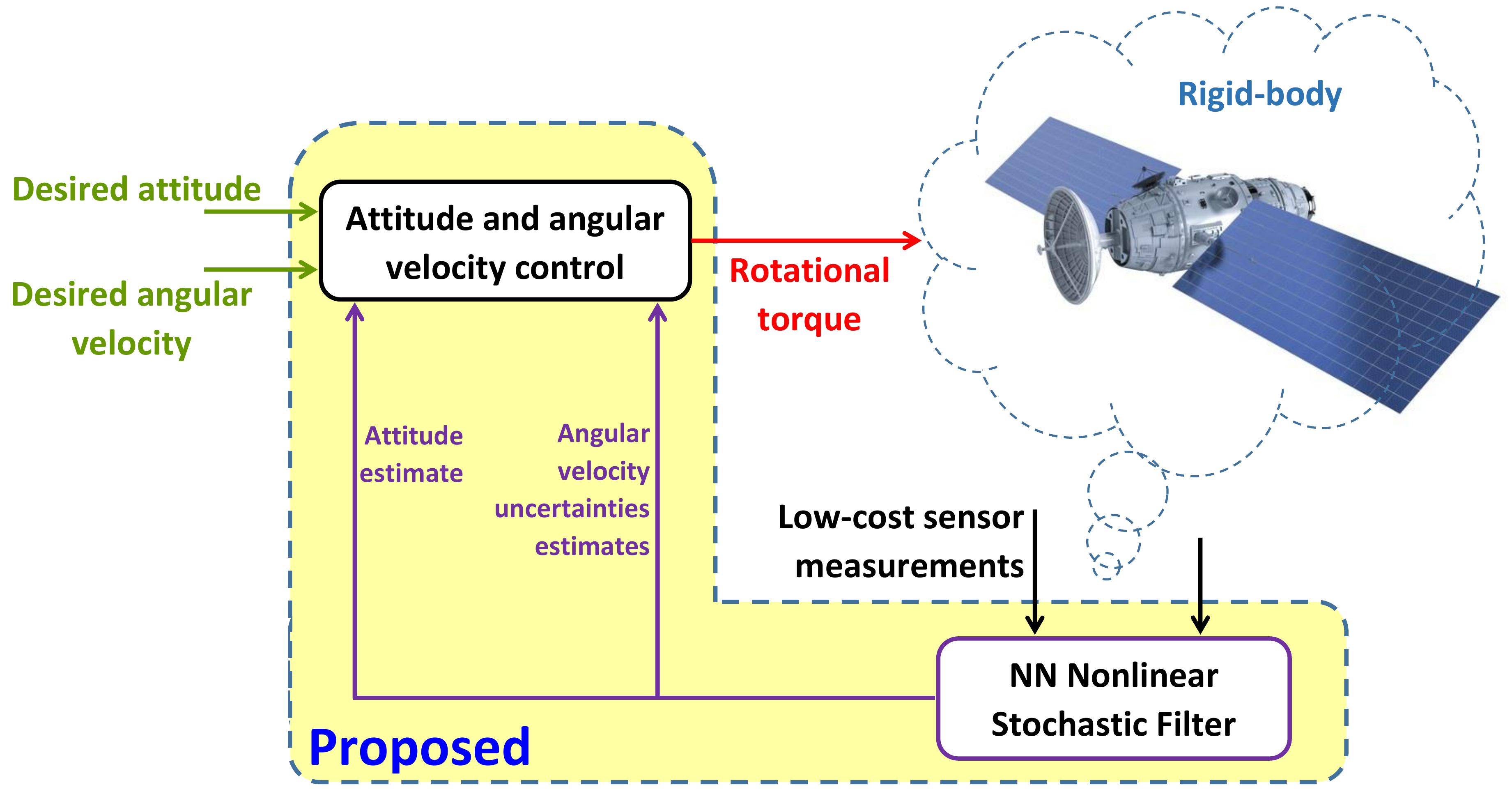}\caption{Illustrative diagram of the proposed filter-based controller for the
		attitude tracking problem.}
	\label{fig:Graph_Summary}
\end{figure}

\section{Neural Network-based Stochastic Filter\label{sec:NN-Filter}}

In this Section, our goal is to design a real-time NN-based nonlinear
stochastic filter for the attitude estimation problem that uses available
measurements directly without attitude reconstruction. Let us define
$\hat{R}\in\mathbb{SO}(3)$ as the estimate of $R$. Let the estimation
error be defined as
\begin{equation}
	\tilde{R}_{o}=R^{\top}\hat{R}\in\mathbb{SO}(3)\label{eq:Attit_Re}
\end{equation}
Let us define $\hat{W}_{b}$ as the estimate of $W_{b}$ in \eqref{eq:Attit_Om_m}.
Define the filter dynamics as follows:
\begin{equation}
	\dot{\hat{R}}=\hat{R}[\Omega_{m}-\hat{W}_{b}-C]_{\times}\label{eq:Attit_Rest_dot}
\end{equation}
where $C\in\mathbb{R}^{3}$ refers to a  correction matrix, $\hat{W}_{b}\in\mathbb{R}^{3\times1}$
being the estimate of $W_{b}$ in \eqref{eq:Attit_Om_m}, and $C$
and $\hat{W}_{b}$ will be designed subsequently. Let the estimation
error between $W_{b}$ and $\hat{W}_{b}$ be as follows:
\begin{align}
	\tilde{W}_{b} & =W_{b}-\hat{W}_{b}\in\mathbb{R}^{3}\label{eq:Attit_Wbe}
\end{align}

\subsection{Direct Measurement Setup}

In view of the vector measurements in \eqref{eq:Attit_Vec_yi}, let
us introduce the following variable:
\begin{equation}
	\hat{{\bf y}}_{i}=\hat{R}^{\top}{\bf r}_{i}\label{eq:Attit_yi_est}
\end{equation}
From \eqref{eq:Attit_Vec_yi_norm}, define the following two variables
\begin{equation}
	\begin{cases}
		M_{r} & =\sum_{i=1}^{N}s_{i}{\bf r}_{i}{\bf r}_{i}^{\top}\\
		M_{y} & =\sum_{i=1}^{N}s_{i}{\bf y}_{i}{\bf y}_{i}^{\top}=\sum_{i=1}^{N}s_{i}R^{\top}{\bf r}_{i}{\bf r}_{i}^{\top}R=R^{\top}M_{r}R
	\end{cases}\label{eq:Attit_Mr_My}
\end{equation}
where $s_{i}$ denotes confidence measure of the $i$th observation/measurement.
For the stability analysis, let us redefine the expression in \eqref{eq:Attit_Vec_yi}
as ${\bf y}_{i}=R^{\top}{\bf r}_{i}$. Consequently, one finds
\begin{align}
	M_{y}\tilde{R}_{o} & =\sum_{i=1}^{N}s_{i}{\bf y}_{i}{\bf y}_{i}^{\top}R^{\top}\hat{R}=\sum_{i=1}^{N}s_{i}{\bf y}_{i}\hat{{\bf y}}_{i}^{\top}\label{eq:Attit_ReMy}
\end{align}
From \eqref{eq:Attit_VEX} and \eqref{eq:Attit_ReMy}, one shows
\begin{align}
	\boldsymbol{\Upsilon}(M_{y}\tilde{R}_{o}) & =\frac{1}{2}\mathbf{vex}(M_{y}\tilde{R}_{o}-\tilde{R}_{o}^{\top}M_{y}^{\top})\nonumber \\
	& =\frac{1}{2}\mathbf{vex}(\sum_{i=1}^{N}{\bf y}_{i}\hat{{\bf y}}_{i}^{\top}-\sum_{i=1}^{N}\hat{{\bf y}}_{i}{\bf y}_{i}^{\top})\nonumber \\
	& =\sum_{i=1}^{n}\frac{s_{i}}{2}\hat{{\bf y}}_{i}\times{\bf y}_{i}\label{eq:Attit_VEX_ReMy}
\end{align}
Likewise, for $||M_{y}\tilde{R}_{o}||_{{\rm I}}=\frac{1}{4}{\rm Tr}\{M_{y}-M_{y}\tilde{R}_{o}\}$
and in view of \eqref{eq:Attit_ReMy}, one obtains
\begin{align}
	||M_{y}\tilde{R}_{o}||_{{\rm I}} & =\frac{1}{4}{\rm Tr}\{\sum_{i=1}^{N}s_{i}{\bf y}_{i}{\bf y}_{i}^{\top}-\sum_{i=1}^{N}s_{i}{\bf y}_{i}\hat{{\bf y}}_{i}^{\top}\}\nonumber \\
	& =\frac{1}{4}{\rm Tr}\{\sum_{i=1}^{N}s_{i}{\bf y}_{i}({\bf y}_{i}-\hat{{\bf y}}_{i})^{\top}\}\label{eq:Attit_RMI_ReMy}
\end{align}

\subsection{Error Dynamics and NN Approximation}

From \eqref{eq:Attit_R_dot}, \eqref{eq:Attit_Wbe}, and \eqref{eq:NAV_Filter1_Detailed},
the error dynamics are as follows: 
\begin{align}
	d\tilde{R}_{o} & =R^{\top}d\hat{R}+dR^{\top}\hat{R}\nonumber \\
	& =(\tilde{R}_{o}[\Omega+\tilde{W}_{b}-C]_{\times}+[\Omega]_{\times}^{\top}\tilde{R}_{o})dt+\tilde{R}_{o}[\mathcal{Q}d\beta]_{\times}\nonumber \\
	& =\tilde{R}_{o}[\Omega]_{\times}-[\Omega]_{\times}\tilde{R}_{o}+\tilde{R}_{o}[\tilde{W}_{b}-C]_{\times}dt+\tilde{R}_{o}[\mathcal{Q}d\beta]_{\times}\label{eq:Attit_dRe}
\end{align}
One can show that
\begin{align*}
	\dot{M}_{y}= & R^{\top}M_{r}R[\Omega]_{\times}-[\Omega]_{\times}R^{\top}M_{r}R\\
	= & M_{y}[\Omega]_{\times}-[\Omega]_{\times}M_{y}
\end{align*}
Let us define $||M_{y}\tilde{R}_{o}||_{{\rm I}}=\frac{1}{4}{\rm Tr}\{M_{y}(\mathbf{I}_{3}-\tilde{R}_{o})\}$.
Based on \eqref{eq:Attit_Identity3} and \eqref{eq:Attit_dRe}, one
finds that the Euclidean distance of \eqref{eq:Attit_dRe} is as follows:
\begin{align}
	d||M_{y}\tilde{R}_{o}||_{{\rm I}}= & -\frac{1}{4}{\rm Tr}\{M_{y}d\tilde{R}_{o}\}\nonumber \\
	= & -\frac{1}{4}{\rm Tr}\{M_{y}\tilde{R}_{o}[(\tilde{W}_{b}-C)dt+\mathcal{Q}d\beta]_{\times}\}\label{eq:Attit_dRe_norm}
\end{align}
with ${\rm Tr}\{\dot{M}_{y}\}={\rm Tr}\{M_{y}[\Omega]_{\times}-[\Omega]_{\times}M_{y}\}=0$
and ${\rm Tr}\{M_{y}\tilde{R}_{o}[\Omega]_{\times}-[\Omega]_{\times}M_{y}\tilde{R}_{o}\}=0$.
In this work, linear in parameter structure of NNs is adopted. For
$x\in\mathbb{R}^{n}$ and a function $f=f(x)\in\mathbb{R}^{m}$, the
linear NN weights structure is approximated as follows:
\begin{equation}
	f=W^{\top}\varphi(x)+\alpha_{f}\label{eq:Attit_fx_NN}
\end{equation}
with $W\in\mathbb{R}^{q\times m}$ being a matrix of synaptic weights,
$\varphi(x)\in\mathbb{R}^{q}$ standing for an activation function,
and $\alpha_{f}\in\mathbb{R}^{m}$ describing an approximated error
vector. The activation function can have high-order connection elements,
for example, Gaussian functions, radial basis functions (RBFs) \cite{zhao2015intelligent},
and sigmoid functions \cite{hashim2022NNFilter,siniscalchi2016adaptation}. Our strategy
is to reach accurate attitude and gyro bias estimation as well as
attenuate gyro noise \textit{stochasticity} effect which, in turn,
will result in accurate estimation of nonlinear attitude dynamics.
NNs have the potential of successfully estimating high-order nonlinear
dynamics \cite{chai2019six,chen2017robust,zheng2017consensus}. Define
$\varphi(\boldsymbol{\Upsilon}_{o})=\varphi(\boldsymbol{\Upsilon}(M_{y}\tilde{R}_{o}))\in\mathbb{R}^{q\times1}$
as an activation function, and consider approximating
\begin{align*}
	\boldsymbol{\Upsilon}_{o}(\tilde{W}_{b}-C) & =(\tilde{W}_{b}-C)\Gamma_{b}^{\top}\varphi(\boldsymbol{\Upsilon}_{o})^{\top}+\alpha_{b}\\
	\mathcal{Q}\boldsymbol{\Upsilon}_{o} & =W_{\sigma}^{\top}\varphi(\boldsymbol{\Upsilon}_{o})+\alpha_{\sigma}
\end{align*}
with $q$ being a positive integer that stands for the number of neurons,
$\Gamma_{b}\in\mathbb{R}^{q\times3}$ referring to a constant matrix,
$W_{\sigma}\in\mathbb{R}^{q\times3}$ denoting unknown NN weighed
matrix to be later adaptively estimated and tuned, $C\in\mathbb{R}^{3\times1}$
being an innovation (correction) term, and $\alpha_{b}\in\mathbb{R}$
and $\alpha_{\sigma}\in\mathbb{R}^{3}$ being the approximated errors.
It is worth noting that $\alpha_{b}\rightarrow0$ and $||\alpha_{\sigma}||\rightarrow0$
as $q\rightarrow\infty$. In view of the nonlinear dynamics in \eqref{eq:Attit_dRe_norm}
and the expression in \eqref{eq:Attit_fx_NN}, one obtains
\begin{align}
	d||M_{y}\tilde{R}_{o}||_{{\rm I}}= & \underbrace{\frac{1}{2}(\varphi(\boldsymbol{\Upsilon}_{o})^{\top}\Gamma_{b}(\tilde{W}_{b}-C)+\alpha_{b})}_{f}dt\nonumber \\
	& +\underbrace{\frac{1}{2}(\varphi(\boldsymbol{\Upsilon}_{o})^{\top}W_{\sigma}+\alpha_{\sigma}^{\top})}_{g}d\beta\label{eq:Attit_dRe_norm_NN}
\end{align}
with $\varphi(\boldsymbol{\Upsilon}_{o})=\varphi(\boldsymbol{\Upsilon}(M_{y}\tilde{R}_{o}))\in\mathbb{R}^{q\times1}$
being an activation function, and $\alpha_{b}\in\mathbb{R}$ and $\alpha_{\sigma}\in\mathbb{R}^{3}$.
Define $\overline{W}_{\sigma}=W_{\sigma}W_{\sigma}^{\top}$ where
$\hat{W}_{\sigma}$ is the estimate of $\overline{W}_{\sigma}$. Let
the estimation error between $\overline{W}_{\sigma}$ and $\hat{W}_{\sigma}$
be defined as follows:
\begin{equation}
	\tilde{W}_{\sigma}=\overline{W}_{\sigma}-\hat{W}_{\sigma}\in\mathbb{R}^{q\times q}\label{eq:Attit_Wse}
\end{equation}

\subsection{Direct NN-based Stochastic Filter}

Consider the following direct real-time NN-based stochastic filter
design:
\begin{equation}
	\begin{cases}
		\dot{\hat{R}} & =\hat{R}[\Omega_{m}-\hat{W}_{b}-C]_{\times}\\
		C & =\left(\Gamma_{b}^{\top}\mathbf{I}_{3}+\frac{\varPsi_{2}}{4\varPsi_{1}}(\Gamma_{b}^{\top}\Gamma_{b})^{-1}\Gamma_{b}^{\top}\hat{W}_{\sigma}\right)\varphi(\boldsymbol{\Upsilon}_{o})\\
		\dot{\hat{W}}_{b} & =\gamma_{b}(\varPsi_{1}\Gamma_{b}\varphi(\boldsymbol{\Upsilon}_{o})-k_{ob}\hat{W}_{b})\\
		\dot{\hat{W}}_{\sigma} & =\frac{\varPsi_{2}}{4}\Gamma_{\sigma}\varphi(\boldsymbol{\Upsilon}_{o})\varphi(\boldsymbol{\Upsilon}_{o})^{\top}-k_{o\sigma}\Gamma_{\sigma}\hat{W}_{\sigma}
	\end{cases}\label{eq:NAV_Filter1_Detailed}
\end{equation}
where
\begin{equation}
	\begin{cases}
		||M_{y}\tilde{R}_{o}||_{{\rm I}} & =\frac{1}{4}{\rm Tr}\{\sum_{i=1}^{N}s_{i}{\bf y}_{i}({\bf y}_{i}-\hat{{\bf y}}_{i})^{\top}\}\\
		\boldsymbol{\Upsilon}_{o} & =\boldsymbol{\Upsilon}(M_{y}\tilde{R}_{o})=\sum_{i=1}^{N}\frac{s_{i}}{2}\hat{{\bf y}}_{i}\times{\bf y}_{i}\\
		\varPsi_{1} & =(||M_{y}\tilde{R}_{o}||_{{\rm I}}+1)\exp(||M_{y}\tilde{R}_{o}||_{{\rm I}})\\
		\varPsi_{2} & =(||M_{y}\tilde{R}_{o}||_{{\rm I}}+2)\exp(||M_{y}\tilde{R}_{o}||_{{\rm I}})
	\end{cases}\label{eq:NAV_Filter1_Auxillary}
\end{equation}
with $k_{ob},k_{o\sigma},\gamma_{b}\in\mathbb{R}$ being positive
constants, $\Gamma_{\sigma}\in\mathbb{R}^{q\times q}$ being a positive
diagonal matrix, $\Gamma_{b}\in\mathbb{R}^{q\times3}$ which is selected
such that $\Gamma_{b}^{\top}\Gamma_{b}$ is positive definite, $q$
denoting neurons number, and $\hat{W}_{b}\in\mathbb{R}^{3\times3}$
and $\hat{W}_{\sigma}\in\mathbb{R}^{q\times q}$ being the estimated
weights of $W_{b}$ and $W_{\sigma}$, respectively. $\varphi(\boldsymbol{\Upsilon}_{o})$
denotes an activation function. It is obvious that $\hat{W}_{\sigma}$
is symmetric provided that $\hat{W}_{\sigma}(0)=\hat{W}_{\sigma}(0)^{\top}$.
\begin{thm}
	\label{thm:Theorem1}Consider the stochastic system in \eqref{eq:Attit_dR_dt}.
	Assume that at least two observations and their respective measurements
	in \eqref{eq:Attit_Vec_yi} are available. Couple the NN-based stochastic
	filter in \eqref{eq:NAV_Filter1_Detailed} directly with the measurements
	in \eqref{eq:Attit_Om_m} ($\Omega_{m}=\Omega+W_{b}+n$) and \eqref{eq:Attit_Vec_yi}
	($y_{i}=R^{\top}r_{i}$ $\forall i=1,2,\ldots,N$). Let $||\tilde{R}_{o}(0)||_{{\rm I}}\neq+1$
	(unstable equilibria). Thus, the closed-loop error signals ($||\tilde{R}_{o}||_{{\rm I}}$,
	$\tilde{W}_{b}$, $\tilde{W}_{\sigma}$) are SGUUB in the mean square.
\end{thm}
\begin{proof}Define a Lyapunov function candidate $\mathcal{U}_{o}=\mathcal{U}_{o}(||M_{y}\tilde{R}_{o}||_{{\rm I}},\tilde{W}_{b},\tilde{W}_{\sigma})$
	such that
	\begin{align}
		\mathcal{U}_{o}= & \frac{1}{2}\exp(||M_{y}\tilde{R}_{o}||_{{\rm I}})||M_{y}\tilde{R}_{o}||_{{\rm I}}+\frac{1}{2\gamma_{b}}\tilde{W}_{b}^{\top}\tilde{W}_{b}\nonumber \\
		& +\frac{1}{2}{\rm Tr}\{\tilde{W}_{\sigma}^{\top}\Gamma_{\sigma}^{-1}\tilde{W}_{\sigma}\}\label{eq:Attit_Lyap}
	\end{align}
	where $\mathcal{U}_{o}:\mathbb{SO}\left(3\right)\times\mathbb{R}^{3}\times\mathbb{R}^{q\times q}\rightarrow\mathbb{R}_{+}$.
	Define $\overline{M}_{y}={\rm Tr}\{M_{y}\}\mathbf{I}_{3}-M_{y}$ and
	recall Lemma \ref{lem:Lem_VEX_RMI}. Since $1\leq\exp(||\tilde{R}_{o}||_{{\rm I}})<3$
	\cite{hashim2019AtiitudeSurvey}, let us define $\underline{\delta}=\inf_{t\geq0}\frac{\underline{\lambda}_{\overline{M}}}{4}$
	where $\inf$ denotes the infimum and $\overline{\delta}=\sup_{t\geq0}\exp(\frac{\overline{\lambda}_{\overline{M}}}{2})$
	with $\sup$ representing the supremum. Hence, one obtains
	\begin{align*}
		e_{o}^{\top}H_{1}e_{o} & \leq\mathcal{U}_{o}\leq e_{o}^{\top}H_{2}e_{o}
	\end{align*}
	where
	\[
	\underline{\lambda}(H_{1})||e_{o}||^{2}\leq\mathcal{U}_{o}\leq\overline{\lambda}(H_{2})||e_{o}||^{2}
	\]
	with $H_{1}=diag(\underline{\delta}^{2},\frac{\gamma_{b}}{2},\frac{1}{2}\underline{\lambda}(\Gamma_{\sigma}^{-1}))$,
	$H_{2}=diag(3\overline{\delta}^{2},\frac{\gamma_{b}}{2},\frac{1}{2}\overline{\lambda}(\Gamma_{\sigma}^{-1}))$,
	and $e_{o}=[\sqrt{||\tilde{R}_{o}||_{{\rm I}}},||\tilde{W}_{b}||,||\tilde{W}_{\sigma}||_{F}]^{\top}$.
	Owing to the fact that $\underline{\delta},\overline{\delta},\gamma_{b},\underline{\lambda}(\Gamma_{\sigma}^{-1}),\overline{\lambda}(\Gamma_{\sigma}^{-1})>0$,
	it becomes apparent that $\underline{\lambda}(H_{1})>0$ and $\overline{\lambda}(H_{2})>0$
	which indicates that $\mathcal{U}_{o}>0$ for all $e_{o}\in\mathbb{R}^{3}\backslash\{0\}$.
	One is able to show that the first $\frac{\partial\mathcal{U}_{o}}{\partial||M_{y}\tilde{R}_{o}||_{{\rm I}}}=\frac{\varPsi_{1}}{2}$
	and second $\frac{\partial^{2}\mathcal{U}_{o}}{\partial||M_{y}\tilde{R}_{o}||_{{\rm I}}^{2}}=\frac{\varPsi_{2}}{2}$
	partial derivatives of $\mathcal{U}_{o}$ relative to $||M_{y}\tilde{R}_{o}||_{{\rm I}}$
	are as follows:
	\begin{equation}
		\begin{cases}
			\varPsi_{1} & =(1+||M_{y}\tilde{R}_{o}||_{{\rm I}})\exp(||M_{y}\tilde{R}_{o}||_{{\rm I}})\\
			\varPsi_{2} & =(2+||M_{y}\tilde{R}_{o}||_{{\rm I}})\exp(||M_{y}\tilde{R}_{o}||_{{\rm I}})
		\end{cases}\label{eq:Attit_Lvv}
	\end{equation}
	Therefore, in view of \eqref{eq:Attit_dRe_norm_NN}, \eqref{eq:Attit_Lyap},
	\eqref{eq:Attit_Lvv}, and Lemma \ref{Lemm:Def_LV_dot}, one obtains
	th following differential operator:
	\begin{align}
		\mathcal{L}\mathcal{U}_{o}= & \varPsi_{1}f+\frac{1}{2}{\rm Tr}\{gg^{\top}\varPsi_{2}\}-\frac{1}{\gamma_{b}}\tilde{W}_{b}^{\top}\dot{\hat{W}}_{b}\nonumber \\
		& -{\rm Tr}\{\tilde{W}_{\sigma}^{\top}\Gamma_{\sigma}^{-1}\dot{\hat{W}}_{\sigma}\}\label{eq:Attit_Lyap_1}
	\end{align}
	From \eqref{eq:NAV_Filter1_Detailed} and \eqref{eq:Attit_Lyap_1},
	one finds
	\begin{align}
		\mathcal{L}\mathcal{U}_{o}= & \varPsi_{1}(\varphi(\boldsymbol{\Upsilon}_{o})^{\top}\Gamma_{b}(\tilde{W}_{b}-C)+\alpha_{b})\nonumber \\
		& +\frac{\varPsi_{2}}{8}{\rm Tr}\{(W_{\sigma}^{\top}\varphi(\boldsymbol{\Upsilon}_{o})+\alpha_{\sigma})(W_{\sigma}^{\top}\varphi(\boldsymbol{\Upsilon}_{o})+\alpha_{\sigma})^{\top}\}\nonumber \\
		& -\frac{1}{\gamma_{b}}\tilde{W}_{b}^{\top}\dot{\hat{W}}_{b}-{\rm Tr}\{\tilde{W}_{\sigma}^{\top}\Gamma_{\sigma}^{-1}\dot{\hat{W}}_{\sigma}\}\nonumber \\
		\leq & \varPsi_{1}\varphi(\boldsymbol{\Upsilon}_{o})^{\top}\Gamma_{b}(\tilde{W}_{b}-C)\nonumber \\
		& +\frac{\varPsi_{2}}{4}{\rm Tr}\{\overline{W}_{\sigma}\varphi(\boldsymbol{\Upsilon}_{o})\varphi(\boldsymbol{\Upsilon}_{o})^{\top}\}-\frac{1}{\gamma_{b}}\tilde{W}_{b}^{\top}\dot{\hat{W}}_{b}\nonumber \\
		& -{\rm Tr}\{\tilde{W}_{\sigma}^{\top}\Gamma_{\sigma}^{-1}\dot{\hat{W}}_{\sigma}\}+\varPsi_{1}\alpha_{b}+\frac{\varPsi_{2}}{4}||\alpha_{\sigma}||^{2}\label{eq:Attit_Lyap_2}
	\end{align}
	where \textit{Young's inequality }has been applied to the following
	expression: $\alpha_{\sigma}^{\top}W_{\sigma}^{\top}\varphi(\boldsymbol{\Upsilon}_{o})=\frac{1}{2}\varphi(\boldsymbol{\Upsilon}_{o})^{\top}\overline{W}_{\sigma}\varphi(\boldsymbol{\Upsilon}_{o})+\frac{1}{2}||\alpha_{\sigma}||^{2}$.
	Let us define $\epsilon_{1}=\sup_{t\geq0}\varPsi_{1}$ and $\epsilon_{2}=\sup_{t\geq0}\varPsi_{2}$.
	From \eqref{eq:Attit_Wbe} and \eqref{eq:Attit_Wse}, substitute $\overline{W}_{\sigma}$
	in \eqref{eq:NAV_Filter1_Detailed} for $\overline{W}_{\sigma}=\tilde{W}_{\sigma}+\hat{W}_{\sigma}$.
	Thereby, utilizing $\dot{\hat{W}}_{b}$, $\dot{\hat{W}}_{\sigma}$,
	and $C$ definitions in \eqref{eq:NAV_Filter1_Detailed}, the result
	in \eqref{eq:Attit_Lyap_2} can be rewritten as follows:
	\begin{align}
		\mathcal{L}\mathcal{U}_{o}\leq & -||\Gamma_{b}^{\top}\varphi(\boldsymbol{\Upsilon}_{o})||^{2}-k_{bo}||\tilde{W}_{b}||^{2}+k_{bo}||\tilde{W}_{b}||\,||W_{b}||\nonumber \\
		& -k_{\sigma o}||\tilde{W}_{\sigma}||_{F}^{2}+k_{\sigma o}||\tilde{W}_{\sigma}||_{F}||W_{\sigma}||_{F}+\epsilon_{1}\alpha_{b}\nonumber \\
		& +\frac{\epsilon_{2}}{4}||\alpha_{\sigma}||^{2}\label{eq:Attit_Lyap_3}
	\end{align}
	On the basis of \textit{Young's inequality} $||\tilde{W}_{b}||\,||W_{b}||\leq\frac{1}{2}||\tilde{W}_{b}||^{2}+\frac{1}{2}||W_{b}||^{2}$
	and $||\tilde{W}_{\sigma}||_{F}||W_{\sigma}||_{F}\leq\frac{1}{2}||\tilde{W}_{\sigma}||_{F}^{2}+||W_{\sigma}||_{F}^{2}$.
	Let us select a hyperbolic tangent activation function $\varphi(\alpha)=\frac{\exp(\alpha)-\exp(-\alpha)}{\exp(\alpha)+\exp(-\alpha)}$
	with $\alpha\in\mathbb{R}$. It becomes apparent that $4||\Gamma_{b}^{\top}\varphi(\boldsymbol{\Upsilon}(\tilde{R}))||^{2}\geq k_{co}||\boldsymbol{\Upsilon}(\tilde{R})||^{2}$
	where $k_{co}=\underline{\lambda}(\Gamma_{b}^{\top}\Gamma_{b})$.
	As such, one finds
	\begin{align}
		\mathcal{L}\mathcal{U}_{o}\leq & -\frac{k_{co}}{4}||\boldsymbol{\Upsilon}_{o}||^{2}-\frac{k_{bo}}{2}||\tilde{W}_{b}||^{2}-\frac{k_{\sigma o}}{2}||\tilde{W}_{\sigma}||_{F}^{2}+\eta_{o}\label{eq:Attit_Lyap_4-1}
	\end{align}
	where $\eta_{o}=\sup_{t\geq0}\frac{k_{bo}}{2}||W_{b}||^{2}+\frac{k_{\sigma o}}{2}||W_{\sigma}||_{F}^{2}+\epsilon_{1}\alpha_{b}+\frac{\epsilon_{2}}{4}||\alpha_{\sigma}||^{2}$.
	Let $\underline{\delta}>1-||\tilde{R}(0)||_{{\rm I}}$ and consider
	Lemma \ref{lem:Lem_VEX_RMI}. Thereby, one obtains
	\begin{align}
		\mathcal{L}\mathcal{U}_{o}\leq & -\frac{\underline{\delta}\,k_{co}}{2}||\tilde{R}_{o}||_{{\rm I}}-\frac{k_{bo}}{2}||\tilde{W}_{b}||^{2}-\frac{k_{\sigma o}}{2}||\tilde{W}_{\sigma}||_{F}^{2}+\eta_{o}\label{eq:Attit_Lyap_4}
	\end{align}
	where $\underline{\delta}=\inf_{t\geq0}\frac{\underline{\lambda}_{\overline{M}}}{4}$.
	As a result, one obtains
	\begin{align}
		\mathcal{L}V\leq & -e_{o}^{\top}\underbrace{\left[\begin{array}{ccc}
				\frac{\underline{\delta}\,k_{co}}{2} & 0 & 0\\
				0 & \frac{k_{bo}}{2} & 0\\
				0 & 0 & \frac{k_{\sigma o}}{2}
			\end{array}\right]}_{H_{3}}e_{o}+\eta_{o}\nonumber \\
		\leq & -\underline{\lambda}(H_{3})||e_{o}||^{2}+\eta_{o}\label{eq:Attit_Lyap_6}
	\end{align}
	with $e_{o}=[\sqrt{||\tilde{R}_{o}||_{{\rm I}}},||\tilde{W}_{b}||,||\tilde{W}_{\sigma}||_{F}]^{\top}$.
	Since $k_{co}>0$, $k_{bo}>0$, and $k_{\sigma o}>0$, it becomes
	obvious that $\underline{\lambda}(H_{3})>0$. Thus, $\mathcal{L}\mathcal{U}_{o}<0$
	if
	\[
	||e_{o}||^{2}>\frac{\eta_{o}}{\underline{\lambda}(H_{3})}
	\]
	Hence, one has
	\begin{equation}
		\frac{d\mathbb{E}[\mathcal{U}_{o}]}{dt}=\mathbb{E}[\mathcal{L}\mathcal{U}_{o}]\leq-\frac{\underline{\lambda}(H_{3})}{\overline{\lambda}(H_{2})}\mathbb{E}[\mathcal{U}_{o}]+\eta_{o}\label{eq:Attit_Lyap7}
	\end{equation}
	Hence, it can be concluded that $e_{o}$ is almost SGUUB completing
	the proof.\end{proof}

\section{Filter-based Controller for Attitude Tracking\label{sec:Controller}}

In this Section, our objective is to design control laws for the attitude
tracking problem reliant on the attitude estimate and direct onboard
measurements such that 1) measurement uncertainties are accounted
for, 2) unknown disturbances are rejected, and 3) overall stability
(interconnection between controller and estimator) is guaranteed.
Define the desired attitude as $R_{d}\in\mathbb{SO}(3)$ and the desired
angular velocity as $\Omega_{d}\in\mathbb{R}^{3}$. Let the error
between the true and the desired attitude be
\begin{equation}
	\tilde{R}_{c}=RR_{d}^{\top}\in\mathbb{SO}(3)\label{eq:Attit_Rec}
\end{equation}
Let the error between the true and the desired angular velocity be
\begin{align}
	\tilde{\Omega}_{c} & =R_{d}^{\top}(\Omega_{d}-\Omega)\in\mathbb{R}^{3}\label{eq:Attit_Om_ec}
\end{align}
From \eqref{eq:Attit_R_dot}, let $d$ be an unknown disturbance attached
to the control input and define $\hat{d}$ as the estimate of $d$.
Let the error between $\hat{d}$ and $d$ be
\begin{equation}
	\tilde{d}=d-\hat{d}\label{eq:Attit_de}
\end{equation}

\begin{assum}\label{Assum:OmD}The desired angular velocity is smooth,
	continuous, and uniformly upper-bounded by a scalar $\gamma_{\Omega}<\infty$
	with $\gamma_{\Omega}\geq\max\{\sup_{t\geq0}||\Omega_{d}||,\sup_{t\geq0}||\dot{\Omega}_{d}||\}$.
	Also, the unknown disturbances are uniformly upper-bounded by a scalar
	$||d||\leq\gamma_{d}<\infty$.\end{assum}

The desired attitude dynamics are given by
\begin{equation}
	\dot{R}_{d}=R_{d}[\Omega_{d}]_{\times}\label{eq:Attit_Rd_dot}
\end{equation}
From \eqref{eq:Attit_Mr_My}, $M_{r}\tilde{R}_{c}=M_{r}RR_{d}^{\top}=\sum_{i=1}^{N}s_{i}{\bf r}_{i}{\bf r}_{i}^{\top}RR_{d}^{\top}$.
Thus, one is able to show that
\begin{align}
	M_{r}\tilde{R}_{c} & =\sum_{i=1}^{N}s_{i}{\bf r}_{i}{\bf y}_{i}^{\top}R_{d}^{\top}\label{eq:Attit_MRec}
\end{align}
with $s_{i}$ denoting the sensor trust level of the $i$th measurement.

\subsection{Control Law from Direct Measurements and Estimated States}

Consider the following control law design:
\begin{equation}
	\begin{cases}
		\boldsymbol{\Upsilon}_{c} & =\boldsymbol{\Upsilon}(M_{r}\tilde{R}_{c})=\sum_{i=1}^{N}s_{i}R_{d}{\bf y}_{i}\times{\bf r}_{i}\\
		\mathcal{T} & =J\dot{\Omega}_{d}-[J(\Omega_{m}-\hat{W}_{b})]_{\times}\Omega_{d}-\hat{d}-w_{c}\\
		w_{c} & =k_{c1}R_{d}^{\top}\boldsymbol{\Upsilon}_{c}+k_{c2}(\Omega_{m}-\hat{W}_{b}-\Omega_{d})\\
		\dot{\hat{d}} & =\frac{k_{d}}{k_{c1}}(\Omega_{m}-\hat{W}_{b}-\Omega_{d})-\gamma_{d}k_{d}\hat{d}
	\end{cases}\label{eq:Attit_Cont_Law}
\end{equation}
where $k_{c1}$, $k_{c2}$, $k_{d}$, and $\gamma_{d}$ are positive
constants, $\hat{W}_{b}$ is the estimate of $W_{b}$, $\hat{d}$
is the estimate of $d$, and $w_{c}$ is an innovation term.
\begin{thm}
	\label{thm:Theorem2}Consider the dynamics in \eqref{eq:Attit_dR_dt}
	with the rotational torque $\mathcal{T}$ being defined as in \eqref{eq:Attit_Cont_Law}.
	Let the control low in \eqref{eq:Attit_Cont_Law} be coupled with
	the filter design in \eqref{eq:NAV_Filter1_Detailed} such that the
	attitude and the unknown weighted bias are estimated using the filter
	design in \eqref{eq:NAV_Filter1_Detailed}. Then, all the closed-loop
	error signals of the filter-based controller ($||\tilde{R}_{o}||_{{\rm I}}$,
	$\tilde{W}_{b}$, $\tilde{W}_{\sigma}$, $||\tilde{R}_{c}||_{{\rm I}}$,
	$\tilde{\Omega}_{c}$, $\tilde{d}$) are SGUUB%
	.
\end{thm}
\begin{proof}Recall the attitude error in \eqref{eq:Attit_Rec},
	the attitude dynamics in \eqref{eq:Attit_R_dot}, and the desired
	attitude dynamics in \eqref{eq:Attit_Rd_dot}. Accordingly, one finds
	the attitude error dynamics as below:
	\begin{align}
		\dot{\tilde{R}}_{c} & =R[\Omega]_{\times}R_{d}^{\top}-R[\Omega_{d}]_{\times}R_{d}^{\top}\nonumber \\
		& =\tilde{R}_{c}[R_{d}(\Omega-\Omega_{d})]_{\times}=\tilde{R}_{c}[\tilde{\Omega}_{c}]_{\times}\label{eq:Attit_Rec_dot}
	\end{align}
	with $\tilde{\Omega}_{c}=R_{d}(\Omega-\Omega_{d})$. Defining $||M_{r}\tilde{R}_{c}||_{{\rm I}}=\frac{1}{4}{\rm Tr}\{M_{r}(\mathbf{I}_{3}-\tilde{R}_{c})\}$
	and recalling \eqref{eq:Attit_dR_dt} and \eqref{eq:Attit_dR_norm},
	one finds
	\begin{align}
		\frac{d}{dt}||M_{r}\tilde{R}_{c}||_{{\rm I}}=\frac{1}{2} & \boldsymbol{\Upsilon}(M_{r}\tilde{R}_{c})^{\top}\tilde{\Omega}_{c}\label{eq:Attit_RMI_Rec_dot}
	\end{align}
	In view of \eqref{eq:Attit_Om_ec}, \eqref{eq:Attit_Cont_Law}, and
	\eqref{eq:Attit_Rec_dot}, one obtains
	\begin{align}
		\frac{d}{dt}JR_{d}^{\top}\tilde{\Omega}_{c}= & [J\Omega]_{\times}R_{d}^{\top}\tilde{\Omega}_{c}+[\Omega_{d}]_{\times}J\tilde{W}_{b}-w_{c}+\tilde{d}\nonumber \\
		& -[Jn]_{\times}\Omega_{d}\label{eq:Attit_Omec_dot}
	\end{align}
	It becomes apparent that
	\begin{align}
		\frac{d}{dt}\tilde{\Omega}_{c}= & -R_{d}(J^{-1}[J\Omega]_{\times}+k_{c2}J^{-1}+[\Omega]_{\times})R_{d}^{\top}\tilde{\Omega}_{c}\nonumber \\
		& +R_{d}^{\top}J^{-1}(k_{c2}\mathbf{I}_{3}-[\Omega_{d}]_{\times}J)\tilde{W}_{b}\nonumber \\
		& -k_{c1}R_{d}J^{-1}R_{d}^{\top}\boldsymbol{\Upsilon}(M_{r}\tilde{R}_{c})+R_{d}J^{-1}\tilde{d}\label{eq:Attit_Omec_dot1}
	\end{align}
	In addition, one finds
	\begin{align}
		\boldsymbol{\Upsilon}(\dot{\tilde{R}}_{c}) & =\frac{1}{2}({\rm Tr}\{\tilde{R}_{c}\}\mathbf{I}_{3}-\tilde{R}_{c})^{\top}\tilde{\Omega}_{c}=\frac{1}{2}\Psi(\tilde{R}_{c})\tilde{\Omega}_{c}\label{eq:Attit_VEXc_dot}
	\end{align}
	Note that in view of Lemma \ref{lem:Lem_VEX_RMI}
	\[
	\underline{\lambda}_{\overline{M}_{r}}^{2}||\tilde{R}_{c}||_{{\rm I}}\leq||\boldsymbol{\Upsilon}(M_{r}\tilde{R}_{c})||^{2}\leq\overline{\lambda}_{\overline{M}_{r}}^{2}||\tilde{R}_{c}||_{{\rm I}}
	\]
	with $\overline{M}_{r}={\rm Tr}\{M_{r}\}\mathbf{I}_{3}-M_{r}$. Define
	$\delta_{c1}$ as a positive constant. Therefore, from \eqref{eq:Attit_Omec_dot1}
	and \eqref{eq:Attit_VEXc_dot}, one has
	\begin{align}
		& \frac{1}{\delta_{c1}}\frac{d}{dt}\boldsymbol{\Upsilon}(\tilde{R}_{c})^{\top}\tilde{\Omega}_{c}\nonumber \\
		& \hspace{1em}=-\frac{1}{\delta_{c1}}(\frac{d}{dt}\boldsymbol{\Upsilon}(\tilde{R}_{c})^{\top})\tilde{\Omega}_{c}-\frac{1}{\delta_{c1}}\boldsymbol{\Upsilon}(\tilde{R}_{c})^{\top}(\frac{d}{dt}R_{d}\tilde{\Omega}_{c})\nonumber \\
		& \hspace{1em}\leq(\frac{c_{c3}}{\delta_{c1}}||\tilde{\Omega}_{c}||+\frac{c_{c3}}{\delta_{c1}}||\tilde{W}_{b}||+\frac{c_{c3}}{\delta_{c1}}||\tilde{d}||)\sqrt{||\tilde{R}_{c}||_{{\rm I}}}\nonumber \\
		& \hspace{2em}-\frac{k_{c1}c_{c2}}{\delta_{c1}}||\tilde{R}_{c}||_{{\rm I}}\label{eq:Attit_VEXcOmec_dot}
	\end{align}
	where $\gamma_{\Omega}\geq\max\{\sup_{t\geq0}||\Omega_{d}||\}$, $c_{c1}=\max\{\sup_{t\geq0}||J^{-1}||_{F},\sup_{t\geq0}||J^{-1}[J\Omega]_{\times}+k_{c2}J^{-1}+[\Omega]_{\times}||_{F},\sup_{t\geq0}||J^{-1}(k_{c2}\mathbf{I}_{3}-[\Omega_{d}]_{\times}J)||_{F}\}$,
	$c_{c2}=\min\{\frac{\underline{\lambda}_{\overline{M}_{r}}^{2}}{\overline{\lambda}_{J}}\}$,
	and $c_{c3}=\max\{\frac{1}{2}+c_{c1}\overline{\lambda}_{\overline{M}_{r}},k_{c2}+\overline{\lambda}_{J}\gamma_{\Omega}\}$.
	Define the following Lyapunov function candidate $\mathcal{U}_{c}=\mathcal{U}_{c}(||M_{r}\tilde{R}_{c}||_{{\rm I}},R_{d}\tilde{\Omega}_{c},\tilde{d})$:
	\begin{align}
		\mathcal{U}_{c}= & 2||M_{r}\tilde{R}_{c}||_{{\rm I}}+\frac{1}{2k_{c1}}\tilde{\Omega}_{c}^{\top}R_{d}^{\top}JR_{d}\tilde{\Omega}_{c}+\frac{1}{\delta_{c1}}\boldsymbol{\Upsilon}(\tilde{R}_{c})^{\top}\tilde{\Omega}_{c}\nonumber \\
		& +\frac{1}{2k_{d}}\tilde{d}^{\top}\tilde{d}\label{eq:Attit_Lyap_c1}
	\end{align}
	where $\mathcal{U}_{c}:\mathbb{SO}\left(3\right)\times\mathbb{R}^{3}\times\mathbb{R}^{3}\rightarrow\mathbb{R}_{+}$.
	Considering Lemma \ref{lem:Lem_VEX_RMI}, one obtains
	\begin{align*}
		& e_{c}^{\top}\underbrace{\left[\begin{array}{ccc}
				\underline{\lambda}_{\overline{M}_{r}} & -\frac{1}{2\delta_{c1}} & 0\\
				-\frac{1}{2\delta_{c1}} & \frac{1}{2k_{c1}} & 0\\
				0 & 0 & \frac{1}{2k_{d}}
			\end{array}\right]}_{H_{4}}e_{c}\\
		& \hspace{2em}\leq\mathcal{U}_{c}\leq e_{c}^{\top}\underbrace{\left[\begin{array}{ccc}
				\overline{\lambda}_{\overline{M}_{r}} & \frac{1}{2\delta_{c1}} & 0\\
				\frac{1}{2\delta_{c1}} & \frac{1}{2k_{c1}} & 0\\
				0 & 0 & \frac{1}{2k_{d}}
			\end{array}\right]}_{H_{5}}e_{c}
	\end{align*}
	where $e_{c}=\left[\sqrt{||\tilde{R}_{c}||_{{\rm I}}},||\tilde{\Omega}_{c}||,||\tilde{d}||\right]^{\top}$.
	It becomes apparent that for $\underline{\lambda}_{\overline{M}_{r}}>\sqrt{\frac{k_{c1}}{2\delta_{c1}^{2}}}$
	and $\overline{\lambda}_{\overline{M}_{r}}>\sqrt{\frac{k_{c1}}{2\delta_{c1}^{2}}}$,
	one finds that $\underline{\lambda}(H_{4}),\overline{\lambda}(H_{5})>0$
	and, in turn $\mathcal{U}_{c}>0$ for all $e_{c}\in\mathbb{R}^{3}\backslash\{0\}$.
	Let us select $\underline{\lambda}_{\overline{M}_{r}}>\sqrt{\frac{k_{c1}}{2\delta_{c1}^{2}}}$.
	Hence, from \eqref{eq:Attit_RMI_Rec_dot} and \eqref{eq:Attit_Omec_dot},
	one finds
	\begin{align}
		\dot{\mathcal{U}}_{c}= & \boldsymbol{\Upsilon}(M_{r}\tilde{R}_{c})^{\top}\tilde{\Omega}_{c}-\frac{1}{k_{d}}\tilde{d}^{\top}\dot{\hat{d}}+\frac{1}{\delta_{c1}}\frac{d}{dt}\boldsymbol{\Upsilon}(\tilde{R}_{c})^{\top}\tilde{\Omega}_{c}\nonumber \\
		& +\frac{1}{k_{c1}}\tilde{\Omega}_{c}^{\top}R_{d}([J\Omega]_{\times}R_{d}^{\top}\tilde{\Omega}_{c}+[\Omega_{d}]_{\times}J\tilde{W}_{b}-w_{c}\nonumber \\
		& \hspace{4em}+\tilde{d}-[Jn]_{\times}\Omega_{d})\label{eq:Attit_Lyap_c2}
	\end{align}
	where $\mathcal{L}\mathcal{U}_{c}=\dot{\mathcal{U}}_{c}$. From \eqref{eq:Attit_Cont_Law},
	one obtains
	\begin{align}
		\mathcal{L}\mathcal{U}_{c}\leq & -\frac{k_{c1}c_{c2}}{\delta_{c1}}||\tilde{R}_{c}||_{{\rm I}}-\frac{2k_{c2}-1}{2k_{c1}}||\tilde{\Omega}_{c}||^{2}-\frac{1}{2}||\tilde{d}||^{2}\nonumber \\
		& +(\frac{c_{c3}}{\delta_{c1}}||\tilde{\Omega}_{c}||+\frac{c_{c3}}{\delta_{c1}}||\tilde{d}|)\sqrt{||\tilde{R}_{c}||_{{\rm I}}}+\frac{1}{k_{c1}}||\tilde{\Omega}_{c}||\,||\tilde{d}||\nonumber \\
		& +(\frac{c_{c3}}{\delta_{c1}}\sqrt{||\tilde{R}_{c}||_{{\rm I}}}+\frac{c_{1}}{k_{c1}}||\tilde{\Omega}_{c}||)||\tilde{W}_{b}||\nonumber \\
		& +\frac{1}{k_{c1}}||k_{c2}n+[Jn]_{\times}\Omega_{d}||+\frac{1}{2}||d||^{2}\label{eq:Attit_Lyap_c3}
	\end{align}
	Let us define $\eta_{c}=\frac{1}{k_{c1}}\sup_{t\geq0}||k_{c2}n+[Jn]_{\times}\Omega_{d}||+\frac{1}{2}||d||^{2}$.
	Thus, the expression in \eqref{eq:Attit_Lyap_c4} becomes as follows:
	\begin{align}
		\mathcal{L}\mathcal{U}_{c}\leq & -e_{c}^{\top}\underbrace{\left[\begin{array}{ccc}
				\frac{k_{c1}c_{c2}}{\delta_{c1}} & -\frac{c_{c3}}{2\delta_{c1}} & -\frac{c_{c3}}{2\delta_{c1}}\\
				-\frac{c_{c3}}{2\delta_{c1}} & \frac{2k_{c2}-1}{2k_{c1}} & \frac{1}{2k_{c1}}\\
				-\frac{c_{c3}}{2\delta_{c1}} & \frac{1}{2k_{c1}} & \frac{1}{2}
			\end{array}\right]}_{H_{6}}e_{c}\nonumber \\
		& +\frac{c_{c3}}{\delta_{c1}}||\tilde{W}_{b}||\sqrt{||\tilde{R}_{c}||_{{\rm I}}}+\frac{c_{1}}{k_{c1}}||\tilde{\Omega}_{c}||\,||\tilde{W}_{b}||+\eta_{c}\label{eq:Attit_Lyap_c4}
	\end{align}
	One is able to show that $\underline{\lambda}(H_{6})>0$ for $\delta_{c1}>\frac{c_{c3}^{2}\left(k_{c1}+2k_{c2}-3\right)}{2k_{c1}c_{c2}(2k_{c2}-3)}$.
	Consider selecting $\delta_{c1}$ such that $\underline{\lambda}(H_{6})>0$,
	and let $\underline{\lambda}_{H_{6}}=\underline{\lambda}(H_{6})$.
	One has
	\begin{align}
		\mathcal{L}\mathcal{U}_{c}\leq & -\underline{\lambda}_{H_{6}}||e_{c}||^{2}+(\frac{c_{c3}}{\delta_{c1}}\sqrt{||\tilde{R}_{c}||_{{\rm I}}}+\frac{c_{1}}{k_{c1}}||\tilde{\Omega}_{c}||)||\tilde{W}_{b}||\nonumber \\
		& +\eta_{c}\label{eq:Attit_Lyap_c5}
	\end{align}
	From \eqref{eq:Attit_Lyap} and \eqref{eq:Attit_Lyap_c1}, define
	the following Lyapunov function candidate:
	\begin{equation}
		\mathcal{U}_{T}=\mathcal{U}_{o}+\mathcal{U}_{c}\label{eq:Attit_Lyap_T}
	\end{equation}
	From \eqref{eq:Attit_Lyap_6} and \eqref{eq:Attit_Lyap_c5}, one obtains
	\begin{align}
		\mathcal{L}\mathcal{U}_{T}\leq & -\underline{\lambda}_{H_{3}}||e_{o}||^{2}-\underline{\lambda}_{H_{6}}||e_{c}||^{2}+\frac{c_{c3}}{\delta_{c1}}||\tilde{W}_{b}||\sqrt{||\tilde{R}_{c}||_{{\rm I}}}\nonumber \\
		& +\frac{c_{1}}{k_{c1}}||\tilde{\Omega}_{c}||\,||\tilde{W}_{b}||+\eta_{o}+\eta_{c}\label{eq:Attit_Lyap_T1}
	\end{align}
	Let $c_{c}=\max\{\frac{c_{c3}}{\delta_{c1}},\frac{c_{1}}{k_{c1}}\}$
	and $\eta_{T}=\eta_{o}+\eta_{c}$. Thus, one shows
	\begin{align}
		\mathcal{L}\mathcal{U}_{T}\leq & -e_{T}^{\top}\underbrace{\left[\begin{array}{cc}
				\underline{\lambda}_{H_{3}} & c_{c}\\
				c_{c} & \underline{\lambda}_{H_{6}}
			\end{array}\right]}_{H_{T}}e_{T}+\eta_{T}\label{eq:Attit_Lyap_T2}
	\end{align}
	where $e_{T}=[||e_{o}||,||e_{c}||]^{\top}$. The expression in \eqref{eq:Attit_Lyap_T1}
	implies that $\underline{\lambda}_{H_{T}}=\underline{\lambda}(H_{T})>0$
	for $\underline{\lambda}_{H_{3}}>\frac{c_{c}^{2}}{\underline{\lambda}_{H_{6}}}$.
	Let us select $\underline{\lambda}_{H_{3}}>\frac{c_{c}^{2}}{\underline{\lambda}_{H_{6}}}$.
	Therefore, $\mathcal{L}\mathcal{U}_{T}<0$ if
	\[
	||e_{T}||^{2}>\frac{\eta_{T}}{\underline{\lambda}(H_{T})}
	\]
	Thus
	\begin{equation}
		\frac{d\mathbb{E}[\mathcal{U}_{T}]}{dt}=\mathbb{E}[\mathcal{L}\mathcal{U}_{T}]\leq-\frac{\underline{\lambda}(H_{T})}{\lambda_{x}}\mathbb{E}[\mathcal{U}_{T}]+\eta_{T}\label{eq:Attit_Lyap7-1}
	\end{equation}
	with $\lambda_{x}=\max\{\overline{\lambda}(H_{2}),\overline{\lambda}(H_{5})\}$.
	Hence, one concludes that $e_{T}$ is almost SGUUB completing the
	proof.\end{proof}

\section{Summary of Implementation\label{sec:Summary-of-Implementation}}

The detailed implementation steps of the discrete NN-based nonlinear
stochastic filter-based controller for the attitude tracking problem
are presented in Algorithm \ref{alg:DiscFilter1}. Note that $\Delta t$
describes a small sampling time.
\begin{algorithm}
	\caption{\label{alg:DiscFilter1}Stochastic filter-based controller algorithm}
	
	\textbf{Initialization}:
	\begin{enumerate}
		\item[{\footnotesize{}1:}] Set $\hat{R}[0]=\hat{R}_{0}\in\mathbb{SO}\left(3\right)$, $\hat{W}_{b}[0]=\hat{W}_{b|0}=0_{3\times1}$,
		$\hat{W}_{\sigma}[0]=\hat{W}_{\sigma|0}=0_{q\times q}$, $\hat{d}[0]=\hat{d}_{0}=0_{3\times1}$,
		$s_{i}\geq0$ for all $i\geq2$, select $\Gamma_{\sigma},k_{ob},k_{o\sigma},k_{c1},k_{c2},k_{d},\gamma_{b},\gamma_{d}>0$,
		$K_{\Upsilon}=\text{rand}(q,3)\in\mathbb{R}^{q\times3}$, $\Gamma_{b}\in\mathbb{R}^{q\times3}$
		where $\underline{\lambda}(\Gamma_{b}^{\top}\Gamma_{b})>0$, and set
		$k=0$.\vspace{1mm}
	\end{enumerate}
	\textbf{while}
	\begin{enumerate}
		\item[{\footnotesize{}2:}] $\begin{cases}
			{\bf r}_{i} & =\frac{r_{i}}{||r_{i}||},\hspace{1em}{\bf y}_{i}=\frac{y_{i}}{||y_{i}||},\hspace{1em}i=1,2,\ldots,N\\
			\hat{{\bf y}}_{i} & =\hat{R}_{k-1}^{\top}{\bf r}_{i}\\
			\boldsymbol{\Upsilon} & =\boldsymbol{\Upsilon}(M_{y}\tilde{R}_{o})=\sum_{i=1}^{N}\frac{s_{i}}{2}\hat{{\bf y}}_{i}\times{\bf y}_{i}\\
			||M_{y}\tilde{R}_{o}||_{{\rm I}} & =\frac{1}{4}{\rm Tr}\{\sum_{i=1}^{N}s_{i}{\bf y}_{i}({\bf y}_{i}-\hat{{\bf y}}_{i})^{\top}\}
		\end{cases}$\vspace{1mm}
		\item[{\footnotesize{}3:}] $\begin{cases}
			\varPsi_{1}= & (||M_{y}\tilde{R}_{o}||_{{\rm I}}+1)\exp(||M_{y}\tilde{R}_{o}||_{{\rm I}})\\
			\varPsi_{2}= & (||M_{y}\tilde{R}_{o}||_{{\rm I}}+2)\exp(||M_{y}\tilde{R}_{o}||_{{\rm I}})
		\end{cases}$\vspace{1mm}
		\item[{\footnotesize{}4:}] $\varphi(\boldsymbol{\Upsilon})=\tanh(K_{\Upsilon}\boldsymbol{\Upsilon})$
		\textcolor{blue}{/{*} hyperbolic tangent activation function {*}/}\vspace{1mm}
		\item[{\footnotesize{}5:}] $\begin{cases}
			\hat{W}_{b|k}= & \hat{W}_{b|k-1}+\Delta t\gamma_{b}(\varPsi_{1}\Gamma_{b}\varphi(\boldsymbol{\Upsilon})-k_{ob}\hat{W}_{b|k-1})\\
			\hat{W}_{\sigma|k}= & \hat{W}_{\sigma|k-1}\\
			& +\Delta t\Gamma_{\sigma}(\frac{\varPsi_{2}}{4}\varphi(\boldsymbol{\Upsilon})\varphi(\boldsymbol{\Upsilon})^{\top}-k_{o\sigma}\hat{W}_{\sigma|k-1})
		\end{cases}$\vspace{1mm}
		\item[{\footnotesize{}6:}] $C=\left(\Gamma_{b}^{\top}\mathbf{I}_{3}+\frac{\varPsi_{2}}{4\varPsi_{1}}(\Gamma_{b}^{\top}\Gamma_{b})^{-1}\Gamma_{b}^{\top}\hat{W}_{\sigma|k}\right)\varphi(\boldsymbol{\Upsilon})$\vspace{1mm}
		\item[] \textcolor{blue}{/{*} angle-axis parameterization {*}/}\vspace{1mm}
		\item[{\footnotesize{}7:}] $\begin{cases}
			\eta & =(\Omega_{m|k}-\hat{W}_{b|k}-C)\Delta t\\
			\beta & =||\eta||,\hspace{1em}u=\eta/||\eta||\\
			\mathcal{R}_{exp} & =\mathbf{I}_{3}+\sin(\beta)[u]_{\times}+(1-\cos(\beta))[u]_{\times}^{2}
		\end{cases}$\vspace{1mm}
		\item[{\footnotesize{}8:}] $\hat{R}_{k}=\hat{R}_{k-1}\mathcal{R}_{exp}$\vspace{1mm}
		\item[{\footnotesize{}8:}] $\begin{cases}
			\boldsymbol{\Upsilon}_{c} & =\boldsymbol{\Upsilon}(M_{r}\tilde{R}_{c})=\sum_{i=1}^{N}s_{i}R_{d}\hat{{\bf y}}_{i}\times{\bf r}_{i}\\
			w_{c} & =k_{c1}R_{d|k}^{\top}\boldsymbol{\Upsilon}_{c}+k_{c2}(\Omega_{m|k}-\Omega_{d|k}-\hat{W}_{b|k})\\
			\hat{d}_{k} & =\hat{d}_{k-1}\\
			& +\Delta t\left(\frac{k_{d}}{k_{c1}}(\Omega_{d|k}-\Omega_{m|k}+\hat{W}_{b|k})-\gamma_{d}k_{d}\hat{d}_{k-1}\right)\\
			\mathcal{T} & =J\dot{\Omega}_{d|k}-[J(\Omega_{m|k}-\hat{W}_{b|k})]_{\times}\Omega_{d|k}-\hat{d}_{k}-w_{c}
		\end{cases}$\vspace{1mm}
		\item[{\footnotesize{}9:}] $k+1\rightarrow k$
	\end{enumerate}
	\textbf{end while}
\end{algorithm}

\section{Numerical Results \label{sec:SE3_Simulations}}

This section demonstrates the effectiveness of the proposed real-time
NN stochastic filter-based controller on $SO(3)$ at a low sampling
rate ($\Delta t=0.01$ sec). The discrete filter detailed in Algorithm
\ref{alg:DiscFilter1} is validated considering large initialization
error, high level of uncertainties, and unknown disturbances. Let
the initial desired attitude be set as $R_{d}[0]=R_{d|0}=\mathbf{I}_{3}\in\mathbb{SO}\left(3\right)$
and the desired angular velocity rate be
\[
\dot{\Omega}_{d}=0.1\left[\begin{array}{c}
	1\sin(0.15t+\frac{\pi}{4})\\
	0.5\sin(0.1t+\frac{\pi}{3})\\
	0.8\cos(0.12t+\frac{\pi}{2})
\end{array}\right]
\]
\begin{figure}[h]
	\centering{}\includegraphics[scale=0.28]{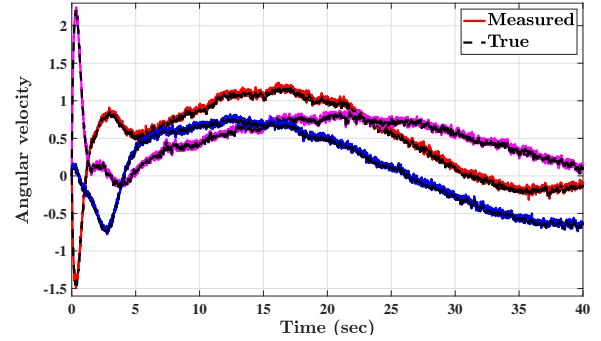}\caption{Sample of angular velocity (0 to 40 seconds): Black center-line (True)
		vs colored solid-line (measured)}
	\label{fig:Fig_Omega}
\end{figure}

\begin{figure}[h]
	\centering{}\includegraphics[scale=0.29]{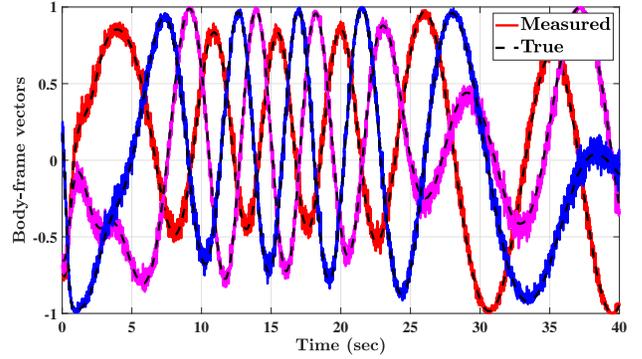}\caption{Sample of vector measurements ($y_{1}$): Black center-line (True)
		vs colored solid-line (measured)}
	\label{fig:Fig_Vect_Meas}
\end{figure}
Let us select the design parameters arbitrary as follows: $k_{ob}=1$,
$k_{o\sigma}=1$, $\gamma_{b}=1$, $\Gamma_{b}=2\mathbf{I}_{3}$,
$\Gamma_{\sigma}=2\mathbf{I}_{3}$, $k_{c1}=10$, $k_{c2}=2$, $k_{d}=10$,
and $\gamma=0.01$. Number of neurons associated with bias and noise
are selected to be both equal to three. Let the initial estimates
of NN weights, disturbances, and attitude be set to $\hat{W}_{b|0}=\hat{W}_{b}[0]=0_{3\times1}$
and $\hat{W}_{\sigma|0}=\hat{W}_{\sigma}[0]=0_{3\times3}$, $\hat{d}_{0}=\hat{d}[0]=0_{3\times1}$,
and $\hat{R}(0)=\hat{R}_{0}=\mathbf{I}_{3}\in\mathbb{SO}\left(3\right)$,
respectively. To validate the robustness of the proposed approach
against high level of uncertainties contributed by the low-cost inertial
measurement unit, consider the measured angular velocity to be corrupted
with unknown weighted bias $W_{b}=0.03[1,0.5,0.7]^{\top}$ $\text{(rad/sec)}$
and normally distributed random noise $n=\mathcal{N}(0,0.05)$ (rad/sec)
(zero mean and standard deviation of $0.05$), see \eqref{eq:Attit_Om_m}.
Consider the following two observations in $\{\mathcal{I}\}$: $r_{1}=\frac{1}{\sqrt{3}}[1,1,-1]^{\top}$
and $r_{2}=[0,0,1]^{\top}$. Consider $\{\mathcal{B}\}$ measurements
to be corrupted by unknown bias $b_{1}=0.05[0.2,0.1,-0.8]^{\top}$
and $b_{2}=0.05[0.5,-0.6,0.4]^{\top}$ as well as normally distributed
random noise $n_{1}=n_{2}=\mathcal{N}(0,0.05)$, visit \eqref{eq:Attit_Vec_yi}.
\begin{figure*}
	\centering{}\includegraphics[scale=0.27]{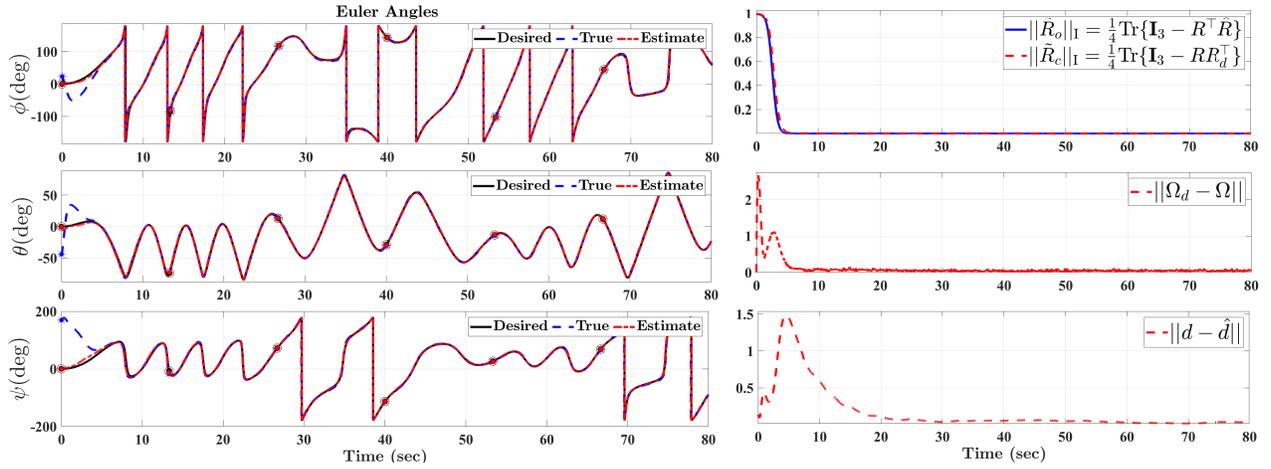}\caption{Right portion demonstrates the evolution trajectories of Euler angles
		(desired marked as a black solid-line, true plotted as a blue dashed-line,
		and estimated depicted as a red center-line). Left portion illustrates
		the error convergence of attitude (estimation error marked as a blue
		solid-line and control error plotted as a red dashed-line), angular
		velocity (red dashed-line), and disturbance (red dashed-line) using
		3 neurons.}
	\label{fig:Fig_Euler}
\end{figure*}
Let the unknown disturbance be $d=0.1[1,3,2]^{\top}$ and the rigid-body's
inertia matrix be $J=diag(0.016,0.015,0.03)$. To test the algorithm
in presence of a very large initialization error in attitude estimation
and control, let the initial value of the true attitude $R$ be 
\[
R[0]=\left[\begin{array}{ccc}
	-0.7060 & 0.0956 & 0.7018\\
	0.1274 & -0.9576 & 0.2585\\
	0.6967 & 0.2719 & 0.6638
\end{array}\right]\in\mathbb{SO}\left(3\right)
\]
and the initial value of the true angular velocity be $\Omega(0)=[0.2,0.3,0.3]^{\top}$
$\text{(rad/sec)}$. Hence, one finds $||\tilde{R}_{o}(0)||_{{\rm I}}=\frac{1}{4}{\rm Tr}\{\mathbf{I}_{3}-R_{0}^{\top}\hat{R}_{0}\}\approx0.999$
and $||\tilde{R}_{c}(0)||_{{\rm I}}=\frac{1}{4}{\rm Tr}\{\mathbf{I}_{3}-R_{0}R_{d|0}^{\top}\}\approx0.999$
very close to the unstable equilibrium $+1$. Let us select hyperbolic
tangent activation function $\varphi(\alpha)=\frac{\exp(\alpha)-\exp(-\alpha)}{\exp(\alpha)+\exp(-\alpha)}$
$\forall$ $\alpha\in\mathbb{R}$ (see Algorithm \ref{alg:DiscFilter1},
step 4).

Fig. \ref{fig:Fig_Omega} contrasts the angular velocity measurements
corrupted by high level of bias and noise with the true values. Likewise,
Fig. \ref{fig:Fig_Vect_Meas} presents high level of uncertainties
corrupting body-frame measurements versus the true data (${\bf y}_{1}$).
The left portion of Fig. \ref{fig:Fig_Euler} depicts the output performance
of the true Euler angles (roll ($\phi$), pitch ($\theta$), yaw ($\psi$)),
the estimated angles ($\hat{\phi}$, $\hat{\theta}$, and $\hat{\psi}$),
and the desired angles ($\phi_{d}$, $\theta_{d}$, $\psi_{d}$).
Fig. \ref{fig:Fig_Euler} shows rapid and accurate tracking performance
of the presented approach. The right portion of Fig. \ref{fig:Fig_Euler}
reveals the robustness of the proposed filter-based controller in
terms of error convergence of: attitude ($||\tilde{R}_{o}||_{{\rm I}}=\frac{1}{4}{\rm Tr}\{\mathbf{I}_{3}-R^{\top}\hat{R}\}$
and $||\tilde{R}_{c}||_{{\rm I}}=\frac{1}{4}{\rm Tr}\{\mathbf{I}_{3}-RR_{d}^{\top}\}$),
angular velocity $||\Omega-\Omega_{d}||$, and disturbance $||d-\hat{d}||$.
As illustrated in Fig \ref{fig:Fig_Euler}, the initially very large
error components converged very close to the origin. Fig. \ref{fig:Fig_West}
presents boundedness of $\hat{W}_{b}$ and $\hat{W}_{\sigma}$ NN
weights plotted with respect to the Euclidean and Frobenius norm,
respectively. Finally, Fig. \ref{fig:Fig_neurons} shows the robustness
of NN approximation in terms of transient response and steady-state
of normalized Euclidean attitude error. Although increasing the number
of neurons results in reduced steady-state error, three neurons where
sufficient to achieve impressive tracking performance. Table \ref{tab:SO3_1}
presents statistical details of mean and standard deviation of the
steady-state error values starting from 4 up to 50 seconds relative
to neuron number.
\begin{figure}[h]
	\centering{}\includegraphics[scale=0.29]{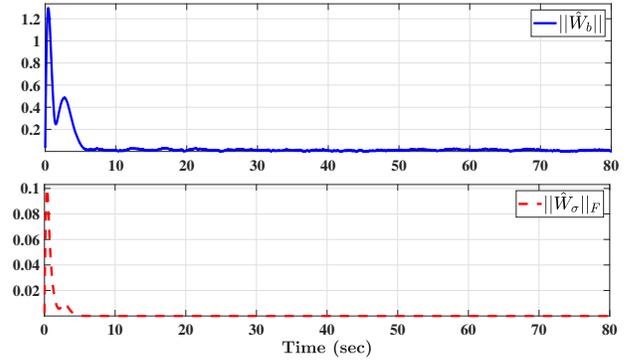}\caption{Boundedness of $\hat{W}_{b}$ and $\hat{W}_{\sigma}$ NN weight estimates
		using 3 neurons.}
	\label{fig:Fig_West}
\end{figure}

\begin{figure}[h]
	\centering{}\includegraphics[scale=0.28]{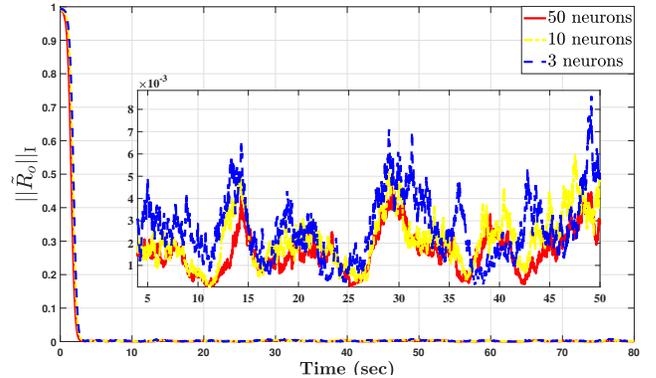}\caption{Normalized Euclidean error $||\tilde{R}_{o}||_{{\rm I}}=\frac{1}{4}{\rm Tr}\{\mathbf{I}_{3}-R_{k}^{\top}\hat{R}_{k}\}$
		considering 3, 10, and 50 neurons.}
	\label{fig:Fig_neurons}
\end{figure}
\begin{table}[H]
	\caption{\label{tab:SO3_1}Steady-state error of $||\tilde{R}_{o}||_{{\rm I}}$
		versus different number of neurons.}
	
	\centering{}%
	\begin{tabular}{c|>{\centering}p{1.8cm}|>{\centering}p{1.8cm}|>{\centering}p{1.8cm}}
		\hline 
		\noalign{\vskip\doublerulesep}
		\multicolumn{4}{c}{Output results of $||\tilde{R}_{o}||_{{\rm I}}=\frac{1}{4}{\rm Tr}\{\mathbf{I}_{3}-R_{k}^{\top}\hat{R}_{k}\}$
			from 4 to 50 sec}\tabularnewline[\doublerulesep]
		\hline 
		\hline 
		\noalign{\vskip\doublerulesep}
		Neurons \# & 3 & 10 & 50\tabularnewline[\doublerulesep]
		\hline 
		\noalign{\vskip\doublerulesep}
		Mean & $2.7\times10^{-3}$ & $2.1\times10^{-3}$ & $1.6\times10^{-3}$\tabularnewline[\doublerulesep]
		\hline 
		\noalign{\vskip\doublerulesep}
		STD & $1.5\times10^{-3}$ & $1.2\times10^{-3}$ & $9.2\times10^{-4}$\tabularnewline[\doublerulesep]
		\hline 
	\end{tabular}
\end{table}
Table \ref{tab:SO3_1} shows that greater number
of neurons leads to better steady-state error convergence.

\section{Conclusion \label{sec:SE3_Conclusion}}

This paper presented a novel neural network (NN) stochastic filter-based
controller for the attitude problem of a rigid-body rotating in three
dimensional space. The proposed approach is posed on the Lie Group
of the Special Orthogonal Group $\mathbb{SO}(3)$. Firstly, an NN-based
stochastic filter design able to operate directly using measurements
supplied by a local unit attached to the rigid-body has been developed.
It has been demonstrated through numerical simulation that the proposed
filter produces accurate estimation given measurements obtained from
a low-cost inertial measurement unit. The proposed filter relies on
online tuning of NN weights which are extracted using Lyapunov stability.
The closed loop error signals of the proposed filter have been shown
to be SGUUB. Next, a novel control law on $\mathbb{SO}(3)$ reliant
on the estimated states and uncertain measurements supplied by a low-cost
onboard measurement unit has been developed. The overall stability
of the filter-based controller has been confirmed and the closed loop
error signals have been shown to be SGUUB. Numerical results illustrate
the robustness and the fast adaptability of the proposed approach.
In the future, the novel neural network (NN) stochastic filter-based
controller can be reformulated to address colored noise uncertainties.

\section*{Acknowledgment}

\subsection*{Appendix\label{subsec:Appendix-A}}
\begin{center}
	\textbf{Neuro-adaptive Filter-based Controller Quaternion Representation}
	\par\end{center}

\noindent Let us define the three-sphere $\mathbb{S}^{3}$ by 
\[
\mathbb{S}^{3}=\{\left.Q\in\mathbb{R}^{4}\right|||Q||=\sqrt{q_{0}^{2}+q^{\top}q}=1\}
\]
with $Q=[q_{0},q^{\top}]^{\top}\in\mathbb{S}^{3}$ being a unit-quaternion
vector composed of two components: $q_{0}\in\mathbb{R}$ and $q\in\mathbb{R}^{3}$.
The equivalent mapping $\mathcal{R}_{Q}:\mathbb{S}^{3}\rightarrow\mathbb{SO}\left(3\right)$
is defined by
\begin{align}
	\mathcal{R}_{Q} & =(q_{0}^{2}-||q||^{2})\mathbf{I}_{3}+2qq^{\top}+2q_{0}\left[q\right]_{\times}\in\mathbb{SO}\left(3\right)\label{eq:NAV_Append_SO3}
\end{align}
see \cite{hashim2019AtiitudeSurvey}. Let $\hat{Q}=[\hat{q}_{0},\hat{q}^{\top}]^{\top}\in\mathbb{S}^{3}$
denote the estimate of $Q=[q_{0},q^{\top}]^{\top}\in\mathbb{S}^{3}$
and $Q_{d}=[q_{d0},q_{d}^{\top}]^{\top}\in\mathbb{S}^{3}$ as the
desired unit-quaternion. The quaternion representation of the NN-based
nonlinear stochastic filter in \eqref{eq:NAV_Filter1_Detailed} is
given below:
\begin{equation}
	\begin{cases}
		C & =\left(\Gamma_{b}^{\top}\mathbf{I}_{3}+\frac{\varPsi_{2}}{4\varPsi_{1}}(\Gamma_{b}^{\top}\Gamma_{b})^{-1}\Gamma_{b}^{\top}\hat{W}_{\sigma}\right)\varphi(\boldsymbol{\Upsilon}_{o})\\
		\dot{\hat{W}}_{b} & =\gamma_{b}(\varPsi_{1}\Gamma_{b}\varphi(\boldsymbol{\Upsilon}_{o})-k_{ob}\hat{W}_{b})\\
		\dot{\hat{W}}_{\sigma} & =\frac{\varPsi_{2}}{4}\Gamma_{\sigma}\varphi(\boldsymbol{\Upsilon}_{o})\varphi(\boldsymbol{\Upsilon}_{o})^{\top}-k_{o\sigma}\Gamma_{\sigma}\hat{W}_{\sigma}\\
		h & =\Omega_{m}-\hat{W}_{b}-C\\
		\mathcal{H} & =\left[\begin{array}{cc}
			0 & -h^{\top}\\
			h & -[h]_{\times}
		\end{array}\right]\\
		\dot{\hat{Q}} & =\frac{1}{2}\mathcal{H}\hat{Q}
	\end{cases}\label{eq:NAV_Filter1_Detailed-1}
\end{equation}
where
\begin{equation}
	\begin{cases}
		||M_{y}\tilde{R}_{o}||_{{\rm I}} & =\frac{1}{4}{\rm Tr}\{\sum_{i=1}^{N}s_{i}{\bf y}_{i}({\bf y}_{i}-\hat{{\bf y}}_{i})^{\top}\}\\
		\boldsymbol{\Upsilon}_{o} & =\sum_{i=1}^{N}\frac{s_{i}}{2}\hat{{\bf y}}_{i}\times{\bf y}_{i}\\
		\varPsi_{1} & =(||M_{y}\tilde{R}_{o}||_{{\rm I}}+1)\exp(||M_{y}\tilde{R}_{o}||_{{\rm I}})\\
		\varPsi_{2} & =(||M_{y}\tilde{R}_{o}||_{{\rm I}}+2)\exp(||M_{y}\tilde{R}_{o}||_{{\rm I}})
	\end{cases}\label{eq:NAV_Filter1_Auxillary-1}
\end{equation}
The equivalent quaternion representation of the control law in \eqref{eq:Attit_Cont_Law}
is given as follows:

\begin{equation}
	\begin{cases}
		\boldsymbol{\Upsilon}_{c} & =\sum_{i=1}^{N}s_{i}\mathcal{R}_{d}{\bf y}_{i}\times{\bf r}_{i}\\
		\mathcal{T} & =J\dot{\Omega}_{d}-[J(\Omega_{m}-\hat{W}_{b})]_{\times}\Omega_{d}-\hat{d}-w_{c}\\
		w_{c} & =k_{c1}\mathcal{R}_{d}^{\top}\boldsymbol{\Upsilon}_{c}+k_{c2}(\Omega_{m}-\hat{W}_{b}-\Omega_{d})\\
		\dot{\hat{d}} & =\frac{k_{d}}{k_{c1}}(\Omega_{m}-\hat{W}_{b}-\Omega_{d})-\gamma_{d}k_{d}\hat{d}
	\end{cases}\label{eq:Attit_Cont_Law-1}
\end{equation}
where $\mathcal{R}_{d}=(q_{d0}^{2}-||q_{d}||^{2})\mathbf{I}_{3}+2q_{d}q_{d}^{\top}+2q_{d0}[q_{d}]_{\times}\in\mathbb{SO}\left(3\right)$.

\bibliographystyle{IEEEtran}
\bibliography{bib_Neuro_Attit}
\begin{IEEEbiography}
	[{\includegraphics[scale=0.11,clip,keepaspectratio]{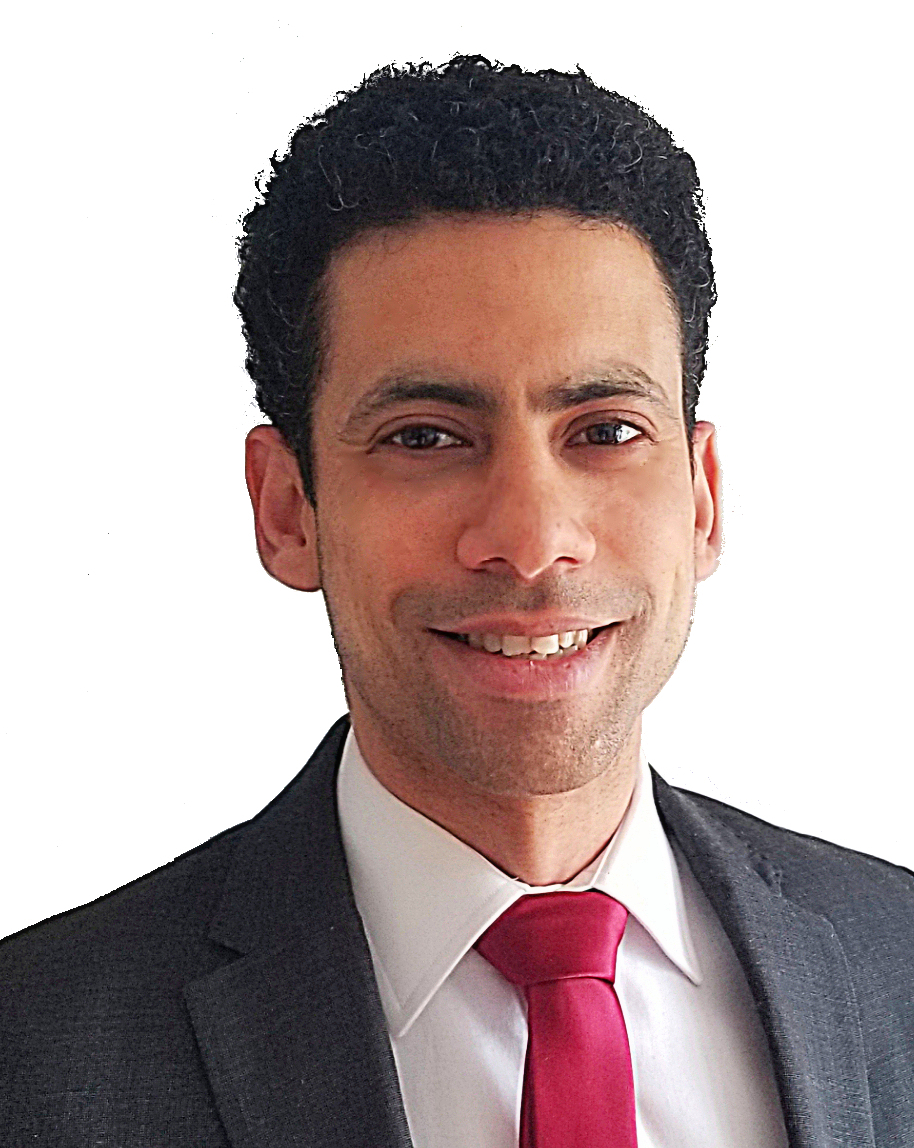}}]{Hashim A. Hashim}
	(Senior Member, IEEE) is an Assistant Professor with the Department of Mechanical and Aerospace
	Engineering, Carleton University, Ottawa, Ontario, Canada. From 2019 to 2021, he was an Assistant Professor	with the Department of Engineering and Applied Science, Thompson Rivers University, Kamloops, British Columbia, Canada.\\
	He received the B.Sc. degree in Mechatronics, Department of Mechanical Engineering from Helwan University, Cairo, Egypt, the M.Sc. in Systems and Control Engineering, Department of Systems Engineering from King Fahd University of Petroleum \& Minerals, Dhahran, Saudi Arabia, and the Ph.D. in Robotics and Control, Department of Electrical and Computer Engineering at Western University, Ontario, Canada.\\
	His current research interests include vision-based aided navigation and control, localization and mapping, stochastic and deterministic estimation, sensor fusion, and distributed control of multi-agent systems.
\end{IEEEbiography}

\begin{IEEEbiography}
	[{\includegraphics[scale=0.21,clip,keepaspectratio]{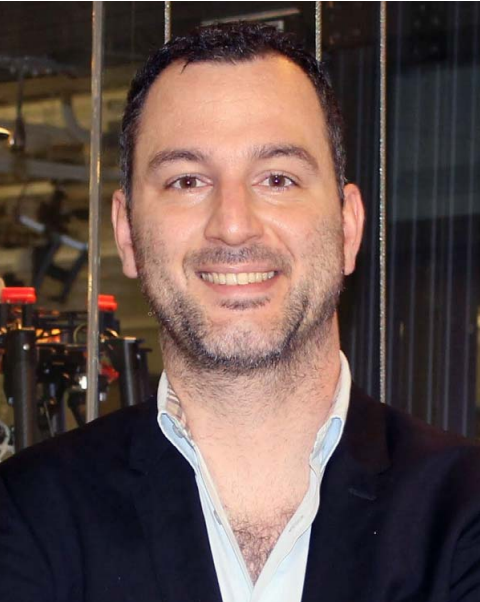}}]{Kyriakos G. Vamvoudakis} (Senior Member, IEEE) was born in Athens, Greece. He received the diploma degree in electronic and computer engineering from the Technical University of Crete, Crete, Greece, in 2006, and the M.Sc. and Ph.D. degrees in electrical engineering from The University of Texas, Arlington, TX, USA, in 2008 and 2011, respectively.\\
	From 2012 to 2016, he was a Research Scientist with the Center for Control, Dynamical Systems and Computation, University of California at Santa Barbara, Santa Barbara, CA, USA. He was an Assistant Professor with the Kevin T. Crofton Department of Aerospace and Ocean Engineering, Virginia Tech, Blacksburg, VA, USA, until 2018. Dr. Vamvoudakis currently serves as an Assistant Professor with The Daniel Guggenheim School of Aerospace Engineering, Georgia Tech, Atlanta, GA, USA. His research interests include reinforcement learning, game theory, cyber-physical security, and safe autonomy.\\
	Dr. Vamvoudakis is a recipient of the 2019 ARO YIP Award, the 2018 NSF CAREER Award, the 2021 GT Chapter Sigma Xi Young Faculty Award, and of several international awards including the 2016 International Neural Network Society Young Investigator (INNS) Award.
\end{IEEEbiography}


\end{document}